\pdfoutput = 1

\documentclass[12pt]{article}
\usepackage{amsmath}
\usepackage{graphicx}
\usepackage{enumerate}
\usepackage{url} % not crucial - just used below for the URL

\usepackage[numbers,sort&compress]{natbib}
\usepackage{calc}
\usepackage{amsthm}
\usepackage{color}
\usepackage{amssymb}
\newcommand{\wh}{\widehat}
\newcommand{\wt}{\widetilde}

\newcommand{\bA}{{\mathbf A}}

\newcommand{\bH}{{\mathbf H}}
\newcommand{\bI}{{\mathbf I}}
\newcommand{\bK}{{\mathbf K}}

\newcommand{\bM}{{\mathbf M}}
\newcommand{\bQ}{{\mathbf Q}}
\newcommand{\bP}{{\mathbf P}}

\newcommand{\bZ}{{\mathbf Z}}

\newcommand{\bff}{{\mathbf f}}

\newcommand{\bOmega}{\boldsymbol{\Omega}}
\newcommand{\bomega}{\boldsymbol{\omega}}

\newcommand{\btheta} {\boldsymbol{\theta}}
\newcommand{\bxi} {\boldsymbol{\xi}}

\newcommand{\bzeta} {\boldsymbol{\zeta}}

\newtheorem{theorem}{Theorem}
\newtheorem{lemma}{Lemma}
\begin{document}

\def\spacingset#1{\renewcommand{\baselinestretch}%
{#1}\small\normalsize} \spacingset{1}

%%%%%%%%%%%%%%%%%%%%%%%%%%%%%%%%%%%%%%%%%%%%%%%%%%%%%%%%%%%%%%%%%%%%%%%%%%%%%%

%\if1\blind
%{
  \title{\bf Model Averaging Estimation for Partially Linear Functional Score Models}
  \author{ Shishi Liu\\
    Center for Applied Statistics, School of Statistics, Renmin University\\
    of China, Beijing, 100872, China\\
    Hao Zhang\\
    School of Statistics, Renmin University of China, Beijing,\\
    100872, China\\
    and \\
    Jingxiao Zhang \thanks{
    Corresponding author. E-mail: zhjxiaoruc@163.com}\hspace{.2cm}\\
    Center for Applied Statistics, School of Statistics, Renmin University\\
    of China, Beijing, 100872, China\\}
  \date{}
  \maketitle
%} \fi

%\if0\blind
%{
%  \bigskip
%  \bigskip
%  \bigskip
%  \begin{center}
%    {\LARGE\bf Model Averaging Estimation for Partially Linear Functional Score Models}
%\end{center}
%  \medskip
%} \fi

\bigskip
\begin{abstract}
This paper is concerned with model averaging estimation for partially linear functional score models. These models predict a scalar response using both parametric effect of scalar predictors and non-parametric effect of a functional predictor. Within this context, we develop a Mallows-type criterion for choosing weights. The resulting model averaging estimator is proved to be asymptotically optimal under certain regularity conditions in terms of achieving the smallest possible squared error loss. Simulation studies demonstrate its superiority or comparability to information criterion score-based model selection and averaging estimators. The proposed procedure is also applied to two real data sets for illustration. That the components of nonparametric part are unobservable leads to a more complicated situation than ordinary partially linear models (PLM) and a different theoretical derivation from those of PLM.
\end{abstract}

\noindent%

%\newpage
\spacingset{1.5} % DON'T change the spacing!

\section{Introduction}
\label{sec1}

Functional data analysis has received growing attention in recent decades, owing to its great flexibility and widespread application in complex data, refer to a comprehensive introduction in \cite{ramsaysilverman2005}. Functional regression models feature prominently in functional data analysis literature, see \cite{morris2015}. A large amount of work has been devoted to regression models with functional predictors, of which the most widely used are functional linear models (FLM). In FLM, a scalar response is associated with the inner product of a functional predictor and an unknown coefficient function, refer to \cite{caihall2006, caiyuan2012, cardotetal1999, cardotetal2003, yaoetal2005a}.
Functional data can be viewed as elements from a functional space such as Hilbert space and reproducing kernel Hilbert space (RKHS). Therefore, dimension reduction is required to address the infinite dimensionality issue in functional data analysis. The popular strategy is to project the functional data into a low-rank functional subspace and take their projections as predictors in regression models. One of the most well-studied dimension reduction tool for functional data is functional principal component analysis (FPCA), discussed in \cite{ricesilverman1991, yaoetal2005b, halletal2006}. Denote $X(t)$ a random function of $L^2(\mathcal{T})$, $t\in\mathcal{T}$, with mean function $\nu(t)$ and covariance function $\mathcal{C}(s,t) = cov\{X(s),X(t)\}$. Classical FPCA take eigen-decomposition of the corresponding covariance operator as $(\mathcal{C} \psi_k)(t) = \lambda_k \psi_k(t)$, $k=1,2,\ldots$, where $\lambda_1 \geq \lambda_2 \geq \cdots$ are eigenvalues and $\{\psi_1(t),\psi_2(t),\ldots\}$ is a set of eigenfunctions. Thus, $X(t)$ has the Karhunen-Lo\`eve expansion
\[
X(t) = \nu(t) + \sum_{k=1}^{\infty}\zeta_k\psi_k(t),
\]
where $\zeta_k = \int_{\mathcal{T}}(X(t)-\nu(t))\psi_k(t)dt$ represents the score associate with the $k$-th eigenfunction, which is called functional principal component (FPC) score.
And researchers use FPC scores associated with leading eigenfunctions as predictors in regression models to specify the effect of functional predictor.

Although widely used, linear models can be restrictive in terms of general applications and many researchers have investigated nonlinear functional regression models, such as \cite{james2002, mullerstadtmuller2005, crainiceanuetal2009, lietal2010, mulleryao2008}.%%%
Some researchers \cite{sangetal2018, wongetal2019} incorporated the effects of both the trajectories and scalar covariates on the prediction of the response.
In these models, the effect of functional predictors is represented by its transformed FPC scores whereas scalar predictors are modeled linearly. The estimation for such models typically truncates the nonparametric part to several leading FPCs, %functions or imposes some regularized penalties on the component functions and estimates the additive components in a RKHS framework,
%see M\"{u}ller and Yao (2008) \cite{Muller and Yao (2008)}, Zhu et al. (2014) \cite{Zhu et al. (2014)}, Sang et al. (2018) \cite{Sang et al. (2018)}, and Wong et al. (2019) \cite{Wong et al. (2019)}).
%Further, an iterative updating procedure is employed to obtain the final estimates for both parts.
see \cite{mulleryao2008, zhuetal2014} also.
The estimation procedure above can be seen as a method for model selection because it results in a parsimonious model by truncating or imposing regularized penalties. Thus, model uncertainty arises from deciding which components are retained in candidate models.

This study considers the partially linear functional score (PLFS) model which describe the connection of a functional predictor and scalar predictors to a scalar response variable of interest, while follows a different estimation strategy namely model averaging. Recall that model selection methods aim to pick out one best model among a set of candidate models, and in this regard, various model selection criteria have been studied, such as Akaike information criterion by \cite{akaike1973} and Bayesian information criterion in \cite{schwarz1978}.
Model averaging, as an alternative to model selection, combines all candidate models by assigning weights to different models to address model uncertainty. Bayesian model averaging has been a popular approach. And Hoeting et~al. (1999) provides a thorough overview of this direction \cite{hoetingetal1999}. %%%
A rapidly growing body of literature with the frequentist paradigm for model averaging has been developed, like \cite{yang2001, yang2003, hansen2007, liangetal2011, hansenracine2012}. The choice of weight plays a fundamental and crucial role in model averaging because it determines the performance of the resulting estimator. Information criterion-based weighting was advocated by \cite{bucklandetal1997, hjortclaeskens2003}, which suggested taking weights based on AIC, BIC, or focused information criterion scores of candidate models. Hansen and Racine (2012) proposed jackknife model averaging \cite{hansenracine2012} and similar weighting procedures based on cross-validation were developed by \cite{chenghansen2015, gaoetal2016}. Both jackknife and cross-validation model averaging may become computationally complex when processing a large sample.
%; hence, efficient algorithms are important when implementing such procedures.
Hansen (2007) proposed Mallows criterion, which suggests weights that minimize this criterion, and established the asymptotic optimality for the model averaging estimator \cite{hansen2007}. Following their work, corresponding Mallows-type %%%
criterions for weight selection in linear mixed-effects model, partially linear model (PLM), and varying-coefficient PLM were established in \cite{zhangetal2014, zhangwang2019, zhuetal2019}, respectively.

To our best knowledge, the literature on functional data contains few works in which the technique of model averaging is applied to regression models. For example, Zhu et~al. (2018) proposed optimal model averaging for partially linear FLM based on Mallows-type criterion \cite{zhuetal2018}. Zhang et~al. (2018), Zhang and Zou (2020) developed a cross-validation model-averaging estimator based on FLM and generalized FLM, respectively \cite{zhangetal2018, zhangzou2020}. In this study, we investigated Mallows-type model averaging for PLFS model. As mentioned above, because of the intrinsically infinite-dimensionality of functional data, a dimension reduction procedure is required, and therefore the components retained for scalar predictors and FPC scores could be expected to have an impact on prediction performance. Because model selection methods pose a risk of selecting an inferior model, we take advantage of model averaging method. This method assigns model weights such that Mallows-type criterion associated with the squared error loss is minimized.

Our work differs from that of \cite{zhangwang2019}, which considered optimal model averaging for PLM. As we use FPC scores to represent the effect of functional predictor, which are unobservable and needed to be estimated first, this situation is more complicated than that of the ordinary PLM. The theoretical derivation of asymptotic optimality for PLFS model is quite different from the previous work owing to the estimated FPC scores. Furthermore, the model uncertainty associated with PLM results from both the choice of covariates and the decision to which part, parametric or nonparametric, the covariate should enter. However, for PLFS model, the uncertainty mainly arises from deciding which scalar covariates and FPC scores should be included in the list of candidates because of the inherent division between scalar predictors and the functional predictor. Besides, this study is also different from \cite{zhuetal2018} in that we handles nonparametric effect whereas they only deal with linear effects. %%%

The remainder of this paper is organized as follows. Section 2 presents model setup and model averaging estimator. The asymptotic optimality of the model averaging estimator is also established in Section 2. Section 3 compares the finite sample performance of the proposed estimator with several information criterion-based model selection and averaging estimators. The proposed procedure is subsequently applied to real data in Section 4. Section 5 concludes our work with a discussion. All proofs are given in Appendix.

\section{Methodology}\label{sec2}

\subsection{Model and estimator}\label{sec21}

Let $Y$ be a scalar response variable associated with a scalar predictor vector $\bZ$ and a functional predictor $X(t)$, $t\in \mathcal{T}$, and let $\{Y_i, \bZ_i, X_i(\cdot)\}_{i=1}^n$ be independent identically distributed (iid) copies of $\{Y, \bZ, X(\cdot)\}$. The relationship between the response and predictors is modeled as
\[
Y_i = m(\bZ_i, X_i) + \varepsilon_i,
\]
where $\varepsilon_i$ are random errors.
Direct modeling $m(\cdot)$ is adversely affected by the ``curse of dimensionality''. Thus, many popular alternatives are developed, such as FLM and functional additive model \cite{mulleryao2008}, modelling the effect of $X(t)$ through its FPC scores based on FPCA. We follow a similar strategy and simplify the modelling as follows.

Denote by $\bzeta_{i}=(\zeta_{i1}, \zeta_{i2},\ldots)$ the sequence of FPC scores of $X_i(t)$ associated with eigenvalues $\{\lambda_1,\lambda_2,\ldots\}$ satisfying $\lambda_1\geq \lambda_2 \geq \cdots > 0$. In addition, $\bxi_i=(\xi_{i1},\xi_{i2},\ldots)^T$ represents the sequence of transformed FPC scores, i.e. $\xi_{ik}=\Phi\big(\lambda_k^{-1/2}\zeta_{ik}\big)$, where $\Phi(\cdot)$ is a continuously differentiable map from $\mathbb{R}$ to $[0,1]$. The transformed FPC scores $\bxi_{i}$ can help to avoid possible scale issues. For simplicity, we take $\Phi(\cdot)$ as a suitable cumulative distribution function (CDF), such as standard Gaussian CDF. If $\zeta_{ik}$ approximately follows Gaussian distribution, $\xi_{ik}$ will almost be uniform in $[0,1]$.
Now consider our PLFS model
\begin{equation}\label{model}
	Y_i = \mu_i + \varepsilon_i = \bZ_i^T\btheta + \bff(\bxi_i) + \varepsilon_i,
\end{equation}
where $\varepsilon=(\varepsilon_1,\ldots,\varepsilon_n)^T$ is random error with conditional mean 0 and variance matrix $\bOmega=\mbox{diag}(\sigma_1^2,\ldots,\sigma_n^2)$.

We use $M$ candidate models to approximate the true PLFS model, where $M$ is allowed to diverge to infinity as $n\to \infty$. The $m$-th candidate PLFS model includes $p_m$ regressors in $\bZ_i$ and $q_m$ regressors in $\bxi_i$ ($m=1,\ldots,M$),
\[
\begin{aligned}
	Y_i &= \mu_{(m),i} + \varepsilon_{(m),i} \\
	&= \bZ_{(m),i}^T\btheta_{(m)} + \bff_{(m)}\big(\bxi_{(m),i}\big) + \varepsilon_{(m),i},
\end{aligned}
\]
where $\bZ_{(m),i}$ is a $p_m\times 1$ vector, $\btheta_{(m)}$ is the corresponding unknown coefficients, $\bxi_{(m),i}$ is a $q_m\times 1$ vector, $\bff_{(m)}$ is an unknown function from $[0,1]^{q_m}$ to $\mathbb{R}$. And $\varepsilon_{(m),i}$ contains the approximation error of the $m$-th candidate and random error.

The kernel smoothing method \cite{speckman1988} is used in estimation.
Denote $\mathcal{K}_{h_m}(\cdot) = \prod_{l=1}^{q_m} k_{h_{m,l}}(\cdot/h_{m,l})$ a product kernel function, where $k_{h_{m,l}}$ is a univariate kernel function and $h_{m,l}$ is the scalar bandwidth. We take $h_{m,l}=h_m$ for simplicity, $l=1,\ldots,q_m$. Furthermore, let $\bK_{(m)} = \big(K_{(m),ij}\big)$ be a $n\times n$ smoother matrix with
\[
K_{(m),ij} = \frac{\mathcal{K}_{h_m}\big(\bxi_{(m),i}-\bxi_{(m),j}\big)}{\sum_{j'=1}^n\mathcal{K}_{h_m}\big(\bxi_{(m),i}-\bxi_{(m),j'}\big)}.
\]
Then, the suggested kernel smoothing estimators of $\btheta_{(m)}$ and $\bff_{(m)}\big(\bxi_{(m)}\big)$ are as follows,
\[
\begin{aligned}
	& \wt{\btheta}_{(m)} = \big(\wt{\bZ}_{(m)}^T\wt{\bZ}_{(m)}\big)^{-1}\wt{\bZ}_{(m)}^T(\bI-\bK_{(m)})Y, \\
	& \wt{\bff}_{(m)}(\bxi_{(m)}) = \bK_{(m)}(Y-\bZ_{(m)}\wt{\btheta}_{(m)}),
\end{aligned}
\]
where $\wt{\bZ}_{(m)} = (\bI-\bK_{(m)})\bZ_{(m)}$. Obviously, $\wt{\btheta}_{(m)}$ is actually a least square estimate and $\wt{\bff}_{(m)}$ is a Nadaraya-Watson estimator.
Therefore, the estimator of $\mu$ under the $m$-th candidate is given by
\[
\begin{aligned}
	\wt{\mu}_{(m)} &= \bZ_{(m)}\wt{\btheta}_{(m)} + \wt{\bff}_{(m)}(\bxi_{(m)}) \\
	&= \wt{\bZ}_{(m)}\big(\wt{\bZ}_{(m)}^T\wt{\bZ}_{(m)}\big)^{-1}\wt{\bZ}_{(m)}^T(\bI-\bK_{(m)})Y + \bK_{(m)}Y \\
	&\equiv \bP_{(m)}Y.
\end{aligned}
\]
Let $\wt{\bP}_{(m)} \equiv \wt{\bZ}_{(m)}\big(\wt{\bZ}_{(m)}^T\wt{\bZ}_{(m)}\big)^{-1}\wt{\bZ}_{(m)}^T$ which is idempotent, and $\bP_{(m)} \equiv \wt{\bP}_{(m)}\big(\bI-\bK_{(m)}\big)+\bK_{(m)}$.

%%%%%%%%%%%%%%%%%%%%%%%%%%%%%%%%%%%%%%%%%%%%%%%%%%%%%%%%%%%%%
\subsection{Weight choice criterion}\label{sec22}

Let $\bomega=(\omega_1, \ldots, \omega_M)^T$ be a weight vector in the unit simplex of $\mathbb{R}^M$,
\[
\mathcal{H}_n = \Big\{\bomega\in [0,1]^M: \sum_{m=1}^M \omega_m = 1\Big\}.
\]
Then the model averaging estimator of $\mu$ follows as
\[
\wt{\mu}(\bomega) = \sum_{m=1}^M \omega_m\wt{\mu}_{(m)} = \sum_{m=1}^M \omega_m\bP_{(m)}Y = \bP(\bomega)Y,
\]
where $\bP(\bomega) = \sum_{m=1}^M \omega_m\bP_{(m)}$. Define the square error loss function and corresponding conditional risk function as
\[
L_n(\bomega) = \|\wt{\mu}(\bomega)-\mu\|^2 = \|\bP(\bomega)Y-\mu\|^2
\]
and
\[
\begin{aligned}
	R_n(\bomega) &= \mathbb{E}(L_n(\bomega)|\bZ, X) \\
	&= \|(\bP(\bomega)-\bI)\mu\|^2 + tr(\bP^T(\bomega)\bP(\bomega)\bOmega),
\end{aligned}
\]
respectively. We may select the optimal weights based on the following Mallows-type criterion
\[
C_n(\bomega) = \|Y-\wt{\mu}(\bomega)\|^2 + 2tr(\bP(\bomega)\bOmega).
\]
It is observed that $\mathbb{E}(C_n(\bomega)|\bZ, X) = R_n(\bomega) + tr(\bOmega)$. Thus, $C_n(\bomega)$ is an unbiased estimator of the expected in-sample squared error loss plus a constant, which is similar to the Mallow's criterion proposed in \cite{hansen2007}. Because $tr(\bOmega)$ is unrelated to $\bomega$, the optimal weights can be obtained by minimizing $C_n(\bomega)$ if $\bOmega$ is known.

However, $\bzeta$ and $\bxi$ are unobservable, so the estimation procedure above cannot be implemented directly. For the sake of practical applicability, we replace the original $\bxi_{(m)}$ with its estimator $\wh{\bxi}_{(m)}$, which is common practice. That is, we first estimate the FPC scores using the previously proposed FPCA method as. That is, suppose the discrete noisy measurements of $X_i(t)$ are available,
\[
X_{ij} = X_i(t_{ij}) + e_{ij},\:i=1,\ldots,n,\:j=1,\ldots,N_i,
\]
where $e_{ij}$'s are independent measurement errors with mean 0 and variance $\sigma_e^2$. We focus on the densely observed trajectories such that $X_i(t)$ can be effectively recovered from $\{(t_{ij}, X_{ij}):j=1,\ldots,N_i\}$ by a smoother operator \cite{kongetal2016,wongetal2019}. The recovered function is denoted by $\wt{X}_i(t)$. Then the mean and covariance functions of $X(t)$ can be estimated by
\[
\begin{aligned}
	\wh{\nu}(t) &= \frac{1}{n}\sum_{i=1}^n\wt{X}_i(t), \\
	\wh{\mathcal{C}}(s,t) &= \frac{1}{n}\sum_{i=1}^n\Big(\wt{X}_i(s)-\wh{\nu}(s)\Big)\Big(\wt{X}_i(t)-\wh{\nu}(t)\Big)^T.
\end{aligned}
\]
The spectral decomposition $\wh{\mathcal{C}}(s,t) = \sum_{k=1}^{n-1} \wh{\lambda}_k\wh{\psi}_k(s)\wh{\psi}_k(t)$ yields sample eigenvalues $\{\wh{\lambda}_k\}$ and eigenfunctions $\{\wh{\psi}_k\}$. The estimates for FPC scores are subsequently obtained by
\[
\begin{aligned}
	\wh{\zeta}_{ik} &= \int_{\mathcal{T}}(\wt{X}_i(t)-\wh{\nu}(t))\wh{\psi}_{k}(t)dt,\\
	\wh{\xi}_{ik} &= \Phi\big(\wh{\lambda}_k^{-1/2}\wh{\zeta}_{ik}\big).
\end{aligned}
\]
Once we get $\wh{\bxi}_{(m)}$, the original quantities listed above have their substitutes in practice, as shown below.

The smoother matrix is now denoted as $\wh{\bK}_{(m)}$ with $i,j$-element
\[
\wh{K}_{(m),ij} = \frac{\mathcal{K}_{h_m}\big(\wh{\bxi}_{(m),i}-\wh{\bxi}_{(m),j}\big)}{\sum_{j'=1}^n\mathcal{K}_{h_m}\big(\wh{\bxi}_{(m),i}-\wh{\bxi}_{(m),j'}\big)}.
\]
The final kernel smoothing estimators of $\btheta_{(m)}$ and $\bff_{(m)}$ are given by
\[
\begin{aligned}
	& \wh{\btheta}_{(m)} = \big(\wh{\bZ}_{(m)}^T\wh{\bZ}_{(m)}\big)^{-1}\wh{\bZ}_{(m)}^T(\bI-\wh{\bK}_{(m)})Y, \\
	& \wh{\bff}_{(m)}(\wh{\bxi}_{(m)}) = \wh{\bK}_{(m)}\big(Y-\bZ_{(m)}\wh{\btheta}_{(m)}\big),
\end{aligned}
\]
where $\wh{\bZ}_{(m)} = \big(\bI - \wh{\bK}_{(m)}\big)\bZ_{(m)}$. Besides, the $m$-th estimator and the model averaging estimator of $\mu$ are
\[
\begin{aligned}
	\wh{\mu}_{(m)} &= \bZ_{(m)}\wh{\btheta}_{(m)} + \wh{\bff}_{(m)}(\wh{\bxi}_{(m)}) \\
	&= \wh{\bZ}_{(m)}\big(\wh{\bZ}_{(m)}^T\wh{\bZ}_{(m)}\big)^{-1}\wh{\bZ}_{(m)}^T\big(\bI-\wh{\bK}_{(m)}\big)Y + \wh{\bK}_{(m)}Y \\
	&\equiv \wh{\bP}_{(m)}Y, \\
	\wh{\mu}(\bomega) &= \sum_{m=1}^M \omega_m\wh{\mu}_{(m)} = \sum_{m=1}^M \omega_m\wh{\bP}_{(m)}Y = \wh{\bP}(\bomega)Y,
\end{aligned}
\]
where $\wh{\bP}(\bomega) = \sum_{m=1}^M \omega_m \wh{\bP}_{(m)}$.

Denote $\overline{\bP}_{(m)} \equiv \wh{\bZ}_{(m)}\big(\wh{\bZ}_{(m)}^T\wh{\bZ}_{(m)}\big)^{-1}\wh{\bZ}_{(m)}^T$ which is still idempotent, and $\wh{\bP}_{(m)} \equiv \overline{\bP}_{(m)}\big(\bI-\wh{\bK}_{(m)}\big)+\wh{\bK}_{(m)}$.
The modified loss, conditional risk, and Mallows-type criterion are transformed into
\[
\begin{aligned}
	\wh{L}_n(\bomega) &= \|\wh{\mu}(\bomega) - \mu\|^2 = \|\wh{\bP}(\bomega)Y - \mu\|^2, \\
	\wh{R}_n(\bomega) &= \mathbb{E}(L_n(\bomega)|\bZ, X), \\
	\wh{C}_n(\bomega) &= \|Y - \wh{\mu}(\bomega)\|^2 + 2tr(\wh{\bP}(\bomega)\bOmega).
\end{aligned}
\]

Let $\wt{\bomega} = \arg\min_{\bomega\in \mathcal{H}_n} \wh{C}_n(\bomega)$. However, the covariance matrix $\bOmega$ is unknown and the criterion $\wh{C}_n(\bomega)$ is therefore still computationally infeasible. Hence, we should estimate $\bOmega$ to obtain a feasible criterion. Following Hansen (2007) \cite{hansen2007}, we estimate $\bOmega$ based on the largest candidate model indexed by $M^{*} = \arg\max_{1\leq m\leq M}(p_m+q_m)$, leading to an estimator
\begin{equation}\label{whOmg}
	\wh{\bOmega} = \mbox{diag}\big(\hat{\epsilon}_{(M^*),1}^2, \ldots, \hat{\epsilon}_{(M^*),n}^2\big),
\end{equation}
where $\big(\hat{\epsilon}_{(M^*),1}, \ldots, \hat{\epsilon}_{(M^*),n}\big)^T = Y - \wh{\mu}_{(M^*)}$.

When $\bOmega$ is replaced by $\wh{\bOmega}$, we select the optimal weights by
\begin{equation}\label{whCn}
	\begin{aligned}
		\wh{\bomega} &= \arg\min_{\bomega} \wh{C}_n(\bomega)|_{\bOmega = \wh{\bOmega}} \\
		&= \arg\min_{\bomega}\|Y-\wh{\mu}(\bomega)\|^2 + 2tr(\wh{\bP}(\bomega)\wh{\bOmega}),
	\end{aligned}
\end{equation}
which can be treated as a feasible counterpart of $\wh{C}_n(\bomega)$.
Let $\bH = \big(Y-\wh{\mu}_{(1)},\ldots,Y-\wh{\mu}_{(M)}\big)$ and $b = \big(tr(\wh{\bP}_{(1)}\wh{\bOmega}),\ldots,tr(\wh{\bP}_{(M)}\wh{\bOmega})\big)^T$. It is clear that (\ref{whCn}) is a standard quadratic programming problem of the form
\[
\begin{aligned}
	& \min_{\bomega} \wh{C}_n(\bomega)|_{\bOmega = \wh{\bOmega}} = \min_{\bomega} \bomega^T\bH^T\bH\bomega + 2\bomega^T b \\
	& \quad \mbox{subject to}\quad \mathbf{1}^T\bomega=1 \: \mbox{and}\: \bomega\geq 0,
\end{aligned}
\]
where $\mathbf{1}$ is a vector with all entries equal to 1. The problem can be efficiently optimized by the R package {\it quadprog}\footnote{S original by Berwin A. \& Turlach R port by Andreas W. (2019). quadprog: Functions to Solve Quadratic Programming Problems. R package version 1.5-7. https://CRAN.R-project.org/package=quadprog.}.

%\footnote{S original by Berwin A. \& Turlach R port by Andreas W. (2019). quadprog: Functions to Solve Quadratic Programming Problems. R package version 1.5-7. https://CRAN.R-project.org/package=quadprog.}.

%%%%%%%%%%%%%%%%%%%%%%%%%%%%%%%%%%%%%%%%%%%%%%%%%%%%%%
\subsection{Asymptotic optimality}\label{sec3}

Define $\eta_n = \inf_{\bomega} R_n(\bomega)$ and let $\lambda_{\max}(\cdot)$ denote the largest singular value of a matrix. Let $\bomega_m^0$ be a weight vector in which the $m$-th component is one and the others are zero. Let $\wt{p} = \max_m p_m$, $\wt{q} = \max_m q_m$, and $h = \min_m h_m$.

The following regularity conditions are required for the model averaging estimator to achieve asymptotic optimality.

(C1).
\[
\begin{aligned}
	& c_{\lambda}^{-1}k^{-\alpha}\leq \lambda_k \leq C_{\lambda}k^{-\alpha}, \\
	& \lambda_k - \lambda_{k+1} \geq C_{\lambda}^{-1}k^{-1-\alpha}, \quad k=1,2,\ldots
\end{aligned}
\]
Assume that $\alpha>1$ to ensure $\sum_{k=1}^{\infty} \lambda_k < \infty$.

(C2). $\mathbb{E}(\|X\|^4)<\infty$ and there exists a constant $C_{\zeta}>0$ such that $\mathbb{E}(\zeta_k^2\zeta_j^2)\leq C_{\xi}\lambda_k\lambda_j$ and $\mathbb{E}(\zeta_k^2-\lambda_k)^2<C_{\zeta}\lambda_k^2$, $\forall k\neq j$.

Condition (C1) assumes that the eigenvalues decay at a polynomial rate, which is a relatively slow rate and allows $X(t)$ to be flexibly modeled as a $L^2$ process. The second part of Condition (C1) requires the spacings among eigenvalues not to be too small to ensure the identifiability and consistency of sample eigenvalues and eigenfunctions. Condition (C1) is widely used in the FLM literature \cite{caihall2006,caiyuan2012}. Condition (C2) is a weak moment restriction on $X(t)$, which is satisfied if $X(t)$ is a Gaussian process. Note that the fourth-order moment condition is commonly used when a good convergence property of the second-order moment is desired. As we use FPC scores to model the effect of $X(t)$ in the proposed method, it is reasonable to use Condition (C2) to ensure the estimated FPC scores are effective.

(C3). The kernel function $k(\cdot)$ is a bounded symmetrical density with compact support and continuous and bounded first derivative function.

(C4). $\max_i \sum_{j=1}^n |K_{(m),ij}| = O(1)$ and $\max_j \sum_{i=1}^n |K_{(m),ij}| = O(1)$ uniformly for $m = 1, \ldots, M$, a.s.

Conditions (C3) and (C4) place certain restrictions upon the kernel method.
Condition (C3) is a common assumption on kernel functions. Condition (C4) bounds the elements of the smoother matrix, which has been discussed in \cite{speckman1988,zhangwang2019}. The smoother matrix constructed by Epanechnikov kernel in simulation study naturally satisfies this condition.

(C5). For some integer $G\geq 1$, $\max_i \mathbb{E}(\varepsilon_i^{4G}|\bZ_i, X_i)<\infty$ for all $i=1,\ldots,n$, a.s.

(C6). $M\eta_n^{-2G}\sum_{m=1}^M \big(R_n(\bomega_m^0)\big)^G = o_p(1)$.

(C7). $\lambda_{\min}\big(\bK_{(m)}\big)\geq c_{K} >0$, where $c_{K}$ is a constant, $m=1,\ldots, M$.
%There exists positive constants $c_1$, $c_2$ such that $c_1\leq \lambda_{\min}\big(\bK_{(m)}\big)\leq \lambda_{\max}\big(\bK_{(m)}\big)\leq c_2$. And $\lambda_{\max}([(\bI-\bK_{(m)})^T(\bI-\bK_{(m)})]^{-1})\lambda_{\max}(\bDelta_{(m)}) < 1$, where $\bDelta_{(m)}=(\bI-\wh{\bK}_{(m)})^T(\bI-\wh{\bK}_{(m)})-(\bI-\bK_{(m)})^T(\bI-\bK_{(m)})$, $m=1,\ldots, M$.

(C8). $\wt{q}=O(n^{1/(2+2\alpha)})$ where $\alpha$ relates to Condition (C1). $n^{-1/2}\wt{q}=o_p(1)$, $n^{1/2}\eta_n^{-1}\wt{q} = o_p(1)$, $\eta_n^{-1}\wt{q}^2 = o_p(1)$.

%$\wt{q}=O(n^{1/(2+2\alpha)})$ where $\alpha$ relates to Condition (C1), $n^{-\frac{1}{2}}\wt{q}=o_p(1)$, $n^{\frac{1}{2}}\eta_n^{-1}\wt{q} = o_p(1)$, $\eta_n^{-1}\wt{q}^2 = o_p(1)$.

(C9). $\|\mu\|^2/n = O(1)$, a.s.

Conditions (C5), (C6) and (C9) are standard conditions for model averaging in the literature. Condition (C5) constrains the conditional moment of random errors, see \cite{hansen2007, zhangetal2014} also. Condition (C6) is commonly used to prove the optimality of model averaging under the scenario that all candidate models are misspecified, which requires $\eta_n$ goes to infinity and constrains rates of the number of candidate models $M$ and the risk of each candidate model; see \cite{wanetal2010, zhangwang2019, zhuetal2019}, among others. Condition (C9) concerns the sum of $\mu_i^2$ and is commonly used in the context of linear regression \cite{liangetal2011}. Condition (C7) is a technical condition in quantifying the order of $\lambda_{\max}\big(\bP_{(m)} - \wh{\bP}_{(m)}\big)$, which requires $\bK_{(m)}$ not being ill-conditioned. Condition (C8) constrains the growth rate of the number of FPC scores, which guarantees an effective estimation accuracy. And we show that under conditions (C7) and (C8), $n\eta_n^{-1}\max_{1\leq m\leq M} \lambda_{\max}\big(\bP_{(m)} - \wh{\bP}_{(m)}\big)=o_p(1)$ holds, which is commonly assumed in the literature \cite{zhangetal2014, zhangyu2018}.

The theorem provides the asymptotic optimality of the model averaging estimator when $\bOmega$ is known.
\begin{theorem}\label{th1}
	Under Conditions (C1)--(C9), it holds that
	\begin{equation}\label{A1}
		\frac{L_n(\wt{\bomega})}{\inf_{\bomega\in \mathcal{H}_n}L_n(\bomega)}\to 1
	\end{equation}
	in probability as $n\to \infty$.
\end{theorem}

Theorem \ref{th1} illustrates the asymptotic optimality of $\wt{\bomega}$ in the sense that the squared loss based on the weight vector $\wt{\bomega}$ is asymptotically identical to that obtained using the infeasible optimal weight vector if $\bOmega$ is known. The proof of Theorem \ref{th1} is shown in the Appendix.

Following \cite{liuokui2013}, we process $tr(\wh{\bP}(\bomega)\bOmega)$ as one entity rather than considering $\bOmega$ in isolation, and estimate it by $\sum_{i=1}^n \hat{\epsilon}_{(M^*),i}^2\rho_{ii}(\bomega)$ where $\rho_{ii}(\bomega)$ is the $i$-th diagonal element of $\wh{\bP}(\bomega)$. Denote $\rho_{ii}^{(m)}$ as the $i$-th diagonal element of $\wh{\bP}_{(m)}$. When $\bOmega$ is replaced by its estimate $\wh{\bOmega}$ given in $(\ref{whOmg})$, provided that the following additional conditions are satisfied, it can be shown that the model averaging estimator based on $\wh{\bomega}$ shares same asymptotic optimality as $\wt{\bomega}$ in Theorem \ref{th1}.

(C10). There exists a constant $c$ such that $\max_{i}\rho_{ii}^{(m)}\leq cn^{-1}|tr\big(\wh{\bP}_{(m)}\big)|$, $\forall m = 1, \ldots, M$.

(C11). $tr\big(\bK_{(m)}\big) = O(h^{-\wt{q}})$ uniformly for $m\in \{1,\ldots,M\}$.

(C12). $\eta_n^{-1}\wt{p} = o_p(1)$ and $\eta_n^{-1}h^{-\wt{q}} = o_p(1)$.

Condition (C10) means that there should not be any dominant or strongly influential subjects as shown in \cite{li1987} and \cite{andrews1991}. Condition (C11) is similar to Condition (h) of \cite{speckman1988} and Condition 4 of \cite{zhangwang2019}. Condition (C12), similar to Condition (C.9) of \cite{zhuetal2019} and Condition 3 of \cite{zhangetal2018}, places additional restrictions on the growth rate of the number of scalar predictors and the number of FPC scores.

\begin{theorem}\label{th2}
	Under Conditions (C1)--(C12), we have that
	\begin{equation}\label{B1}
		\frac{L_n(\wh{\bomega})}{\inf_{\bomega\in\mathcal{H}_n} L_n(\bomega)}\to 1
	\end{equation}
	in probability as $n\to \infty$.
\end{theorem}

Theorem \ref{th2} shows that Theorem \ref{th1} remains valid when $\bOmega$ is replaced by $\wh{\bOmega}$. Thus, the practically feasible $\wh{\bomega}$ also enjoys the asymptotic optimality. The Appendix provides the detailed proof of Theorem \ref{th2}.

\section{Simulation study}\label{sec4}

In this section, we compare the finite sample performance of the proposed Mallows-type model averaging (MMA) estimator with several model selection and averaging estimators based on information criteria.

The data are generated from the following PLFS model,
\begin{equation}\label{mod:1}
	Y_i = \mu_i + \varepsilon_i = \sum_{j=1}^{M_0}\theta_j Z_{ij} + \bff(\bxi_i) + \varepsilon_i, \: i = 1, \ldots, n,
\end{equation}
where $\bxi_i$ is transformed FPC scores vector from $\bzeta_i$ with $\zeta_{ik}$ being generated independently from $N(0, \lambda_k)$ and standard Gaussian CDF being the transformation $\Phi(\cdot)$. The following scenarios are considered.

\noindent\textbf{Design 1.} $M_0 = 50$ and $\theta_j = j^{-2/3}$. $\bZ_i\sim MN(0, \Sigma)$, generated independently from the functional predictor $X_i(t)$, with the $a,b$-th element $\Sigma_{ab}$ being $0.5^{|a-b|}$. $X_i(t)$ is obtained by
\[
X_i(t) = \sum_{k=1}^{40} \zeta_{ik}\psi_k(t),\quad t\in [0,1],
\]
where $\zeta_{ik}\sim N(0, k^{-3/2})$, $\psi_k(t) = \sqrt{2}\sin(k\pi t)$, $k = 1,\ldots,40$. $\varepsilon_i$'s are homoscedastic as $\varepsilon_i\sim N(0, \eta^2)$. Varying $\eta$ such that $R^2 = var(\mu_i)/ var(Y_i)$ varies between 0.1 and 0.9, where $var(\mu_i)$ and $var(Y_i)$ are variances of $\mu_i$ and $Y_i$ respectively. And
\[
\bff(\bxi) = \exp\Big\{\sum_{k=1}^{40}\xi_{k}/k\Big\}.
\]

\begin{figure}[htbp]
	\centering
	\includegraphics[width=0.8\textwidth]{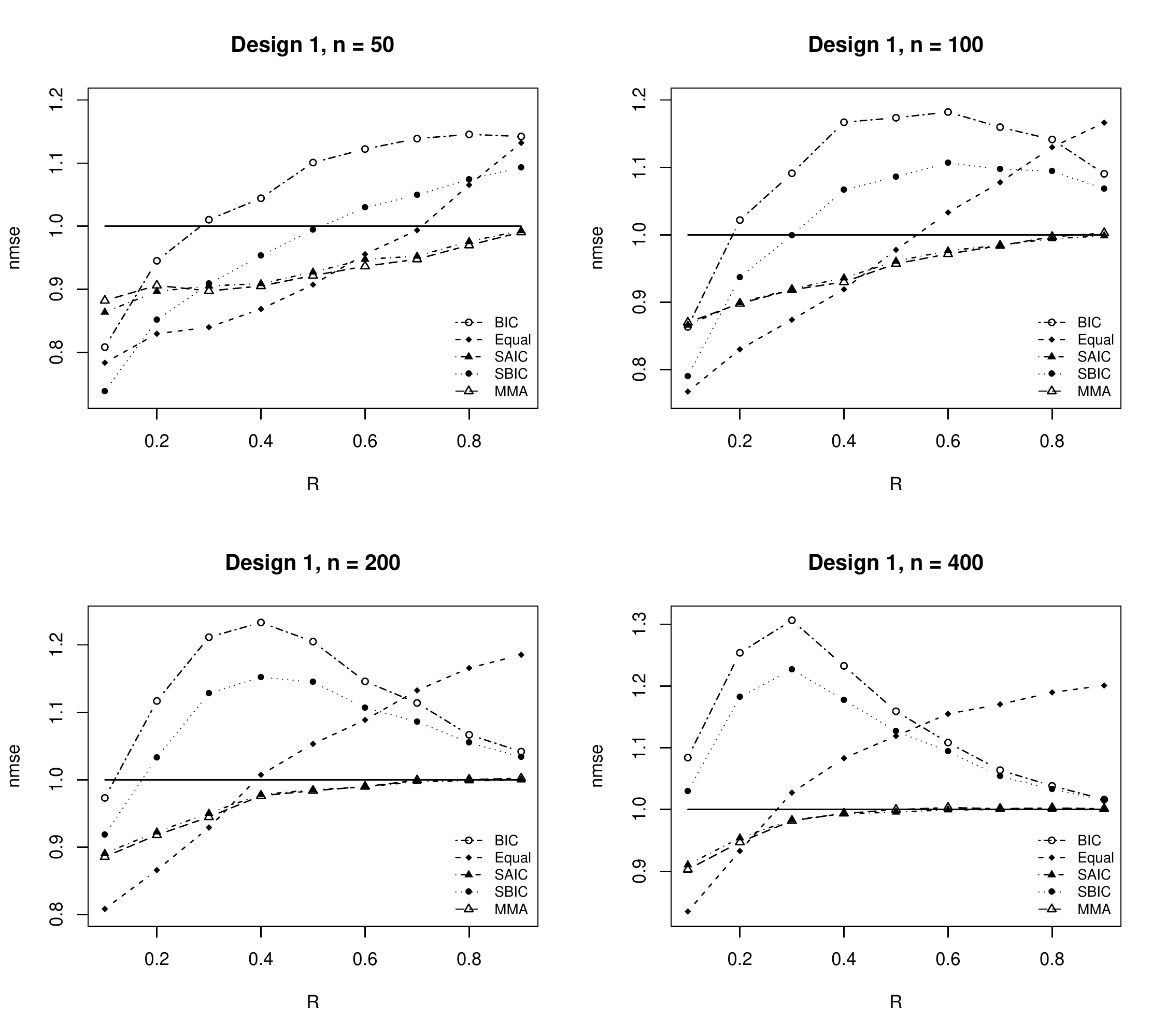}
	\caption{ Normalized mean squared error (NMSE) comparison for Design 1. }
	\label{nmse1}
\end{figure}

\noindent\textbf{Design 2.} $M_0 = 50$ and $\theta_j = j^{-1/2}$. Consider $\bZ$ and $X(t)$ being correlated. Simulate $(\bZ_i, \zeta_{i1})\sim MN(0, \Sigma)$ where the $a,b$-th element $\Sigma_{ab}=0.5^{|a-b|}$. The functional predictor $X_i(t)$ is obtained by
\[
X_i(t) = \sum_{k=1}^{20} \zeta_{ik}\psi_k(t),\quad t\in [0,10],
\]
where $\zeta_{ik}\sim N(0, k^{-2})$, $k = 2,\ldots,20$. $\psi_k(t) = \cos(k\pi t/5)/\sqrt{5}$, $k = 1,\ldots,20$. $\varepsilon_i$'s are heteroscedastic as $\varepsilon_i\sim N\big(0, \eta^2(u_{i}^2+0.01)\big)$, where $u_i$ is $U[-1,1]$. Still varying $\eta$ such that $R^2$ varies between 0.1 and 0.9. And
\[
\bff(\bxi) = \xi_1\xi_2 + \xi_3^2 + \sum_{k=4}^{20} \frac{1}{k}\Big(\xi_k-\frac{1}{2}\Big).
\]

\noindent\textbf{Design 3.} Design 3 is close to Design 1 except that $\bZ_i$ is correlated to $X_i(t)$ as Design 2, and random errors are heteroscedastic as $\varepsilon_i\sim N\big(0, \eta^2(Z_{i1}^2+0.01)\big)$. Still varying $\eta$ such that $R^2$ varies between 0.1 and 0.9.

For each design, $X(t)$ is observed at 100 equally-spaced grids on $\mathcal{T}$ with measurement errors. Denote the $i$-th observation of $X$ at $t_{j}$ by $X_{ij} = X_i(t_j) + e_{ij}$, where measurement errors $e_{ij}$'s are independent $N(0,0.2)$ variables.
The sample size is set to $n=50$, $100$, $200$ and $400$. Consider three kinds of candidate model setting corresponding to each design as follows.

\noindent \textbf{M15a}. For Design 1, a nested setting is considered, that is, containing the first $s$ components of $\bZ$ and $\bxi$. A candidate model contains at least one of $\{Z_1,\ldots,Z_5\}$ and at least one of $\{\xi_1,\xi_2,\xi_3\}$, which leads to $M = 5\cdot 3 = 15$ candidates.

\noindent \textbf{M15b}. For Design 2, we only restrict the nested mode in nonparametric part. The parametric part contains at least one of $\{Z_1,Z_2\}$. For nonparametric part, the first $s$ transformed FPC scores of $\{\xi_1,\ldots,\xi_5\}$ are contained. It results in $ M = \left[\binom{2}{2}+\binom{2}{1}\right]\cdot 5 = 15
$ candidates.

\noindent \textbf{M21}. For Design 3, assume at least one of $\{Z_1,Z_2,Z_3\}$ and at least one of $\{\xi_1,\xi_2\}$ are included in a candidate model. Thus, there are $M = \left[\binom{3}{3}+\binom{3}{2}+\binom{3}{1}\right]\cdot \left[\binom{2}{2}+\binom{2}{1}\right] = 21$ candidates.

\begin{figure}[htbp]
	\centering
	\includegraphics[width=0.8\textwidth]{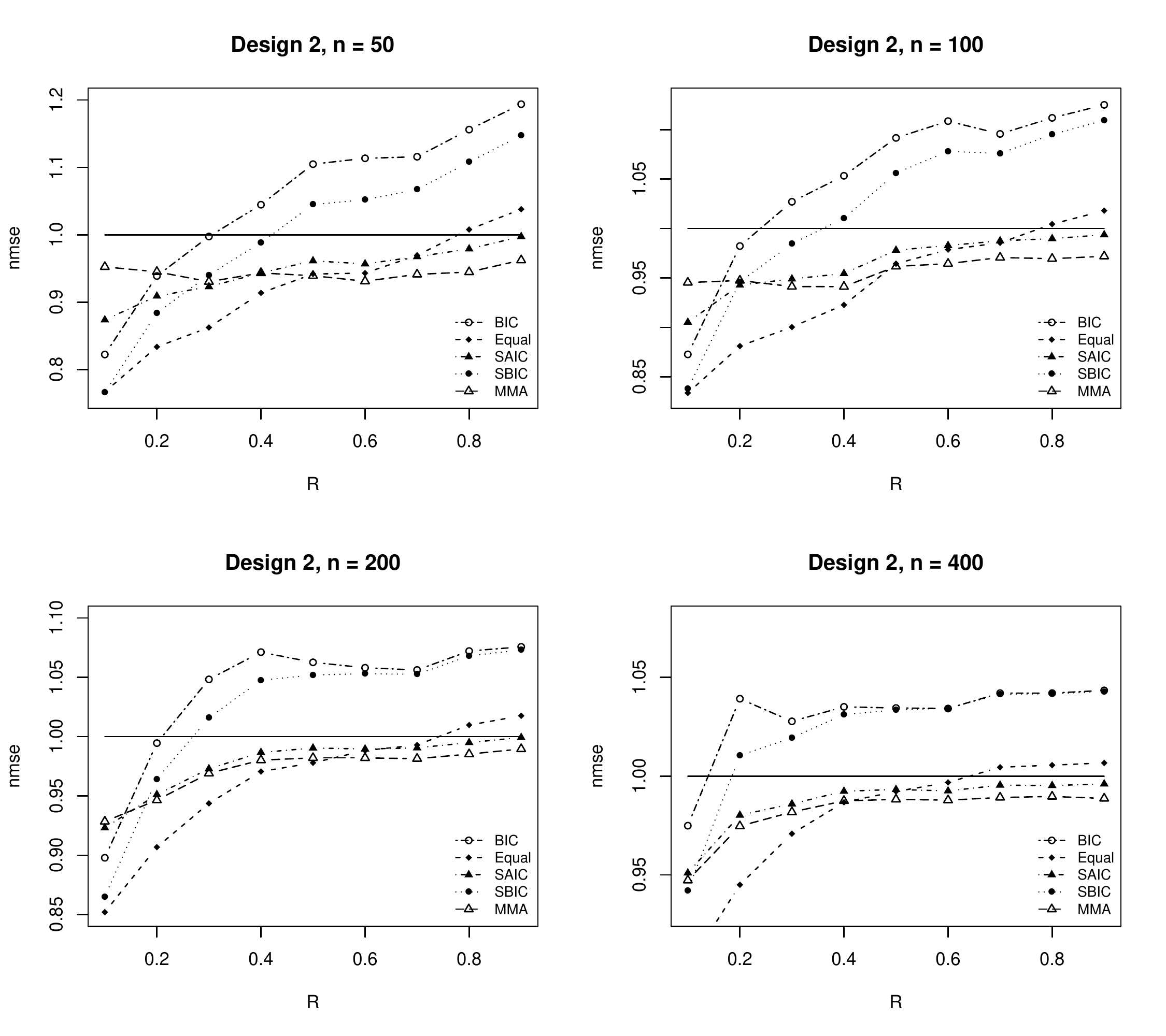}
	\caption{ Normalized mean squared error (NMSE) comparison for Design 2. }
	\label{nmse2}
\end{figure}

The construction of \textbf{M21} is based on the observation that the effects of FPCs on the response does not necessarily coincide with their magnitudes \cite{zhuetal2007,hadiling1998, bairetal2006, zhuetal2014}. Therefore, the structure regarding FPC scores in \textbf{M21} is not restricted to the nested one.

We use Epanechnikov kernel $k(u) = \frac{3}{4}(1-u^2)I(|u|\leq 1)$ for all candidate models with bandwidth $h_m$ being $n^{-1/(1+q_m)}$ based on rule-of-thumb method, $m=1,\ldots,M$. In addition, we compare the finite sample performance of MMA estimator with five alternative methods - AIC, BIC, equally weighting, smoothed AIC (SAIC) and smoothed BIC (SBIC) suggested by Buckland et~al. (1997) \cite{bucklandetal1997}. For the $m$-th candidate model,
AIC and BIC select the model with the smallest scores, defined as $AIC_m = \log(\wh{\sigma}_m^2) + 2tr(\wh{\bP}_{(m)})/n$ and $BIC_m = \log(\wh{\sigma}_m^2) + \log(n)tr(\wh{\bP}_{(m)})/n$, where $\wh{\sigma}_m^2 = \frac{1}{n}\|Y-\wh{\mu}_{(m)}\|^2$.
SAIC and SBIC assign weights to the $m$-th candidate as $\omega_m^{AIC} = \exp(-AIC_m/2) / \sum_{m=1}^M \exp(-AIC_m/2)$ and $\omega_m^{BIC} = \exp(-BIC_m/2) / \sum_{m=1}^M \exp(-BIC_m/2)$, respectively. Equally weighting just assigns equal weights to each candidates. Mean squared error (MSE) of each methods is compared,
\[
MSE = \frac{1}{nD}\sum_{d=1}^{D}\|\wh{\mu}^{(d)}-\mu^{(d)}\|^2,
\]
where $D = 200$ denotes the number of repetitions and $d$ represents the $d$-th trial. For easy comparison, all MSE's are normalized by dividing by the MSE of AIC model selection estimator. Thus, a normalized MSE (NMSE) smaller than 1 indicates the corresponding estimator is superior to AIC estimator, and vice versa.

\begin{figure}[htbp]
	\centering
	\includegraphics[width=0.8\textwidth]{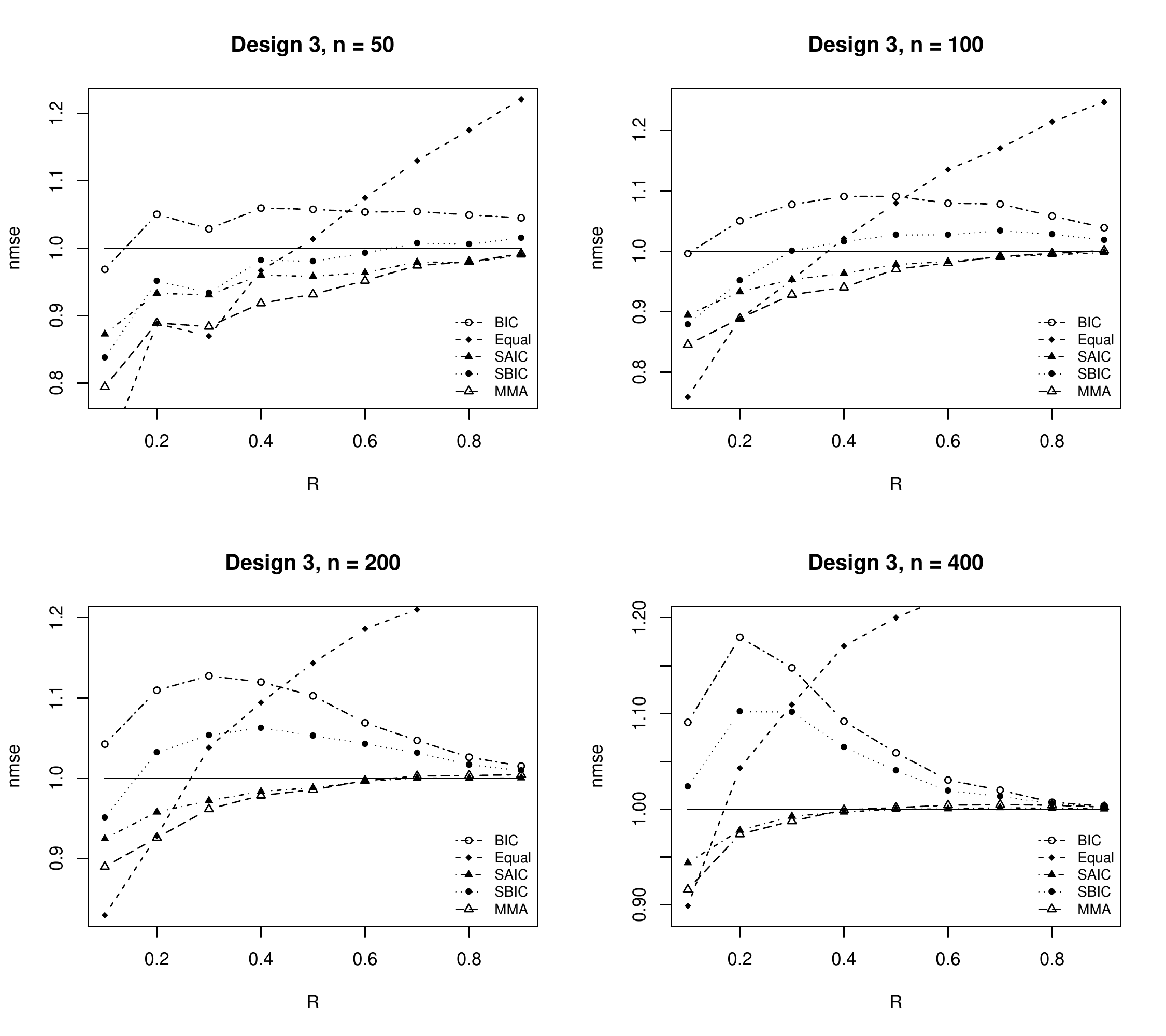}
	\caption{ Normalized mean squared error (NMSE) comparison for Design 3. }
	\label{nmse3}
\end{figure}

%The results of mean MSE and standard error are listed in Supplementary material.
Figures~\ref{nmse1}--\ref{nmse3} present corresponding results of Design 1--3.
For Design 1, MMA and SAIC in Figure~\ref{nmse1} exhibits superiorities for large and medium $R^2$ values, whereas for small $R^2$ value, equally weighting performs the best. BIC and SBIC cannot provide comparable results in Design 1. Besides, MMA performs slightly better than SAIC for small sample size or small $R^2$ values. Also, it is shown in Figure~\ref{nmse1} that SAIC and SBIC outperform their model selection counterparts --- AIC and BIC, where the differences decrease as $R^2$ or $n$ grows.

For Design 2, as shown in Figure~\ref{nmse2}, MMA dominates the other methods for large and medium $R^2$ values. Similar to the results for Design 1, we can also observe that BIC, SBIC and equally weighting have a marginal advantage for small $R^2$ values. As $R^2$ increases, the differences between these six methods decrease.

For Design 3, Figure~\ref{nmse3} illustrates that MMA shows an edge over AIC, BIC, SAIC and SBIC for small and medium $R^2$ values, and the differences between these five methods become smaller as $R^2$ grows. Equllay weighting outperforms the other methods for small $R^2$ values but deteriorates rapidly as $R^2$ grows. MMA, AIC and SAIC behave more closely as $n$ grows.

In summary, the proposed MMA estimator delivers more satisfactory outcomes than the other competing estimators in most cases. The superior performance of MMA estimator in finite sample is partly attributed to the merit that its optimality does not depend on the correct specification of candidate models, which means that the true model is not necessarily included in the candidate set.
Moreover, equally weighting method performs well for small and $R^2$ values and SBIC usually yields rather competitive results when $R^2$ is small, whereas their performances generally worsen as $R^2$ increases. This shows that equally weighting and SBIC methods are not capable in selecting optimal weights with minor noise level in our simulation settings.
In addition, model averaging estimators, SAIC and SBIC, always outperform their model selection counterparts --- AIC and BIC, and the differences between AIC and SAIC, or BIC and SBIC generally decrease as $R^2$ increases. SAIC typically shows a moderate advantage over AIC in most cases.
Furthermore, it is observed that MMA has a growing edge over other methods when the structure of candidate model becomes unrestricted. This directly reveals that the optimality of MMA in finite sample relies on candidate models on hand. Therefore, combining various types of candidate model is appropriate practice.

%\addtolength{\textheight}{.5in}%
%\clearpage

\section{Application to real data}\label{sec5}

In this section, we illustrate the application of the proposed method to two data sets, both consisting of near-infrared (NIR) spectra data and some reference variables.

\subsection{NIR shootout 2002 data set}
The {\it NIR shootout 2002} data set was published by the International Diffuse Reflectance Conference (IDRC) in 2002 and is available from Eigenvetor Research Inc, USA\footnote{Conference International Diffure Reflectance. (2002). NIR Spectra of Pharmaceutical Tablets from Shootout. {\it Eigenvector Research}. http://
	www.eigenvector.com/data/tablets/index.html}. It contains NIR spectra for 655 pharmaceutical tablets (functional predictor $X$), measured at two spectrometers over the spectral region from 600 to 1898 nm with 2 nm increments on the wavelength scale. Some quantities for reference analysis, such as weight of the active ingredient (response variable $Y$), weight of each tablet (scalar predictor $Z_1$), and hardness of each tablet (scalar predictor $Z_2$), are also provided. The data have already been divided into training (155), validation (40) and test (460) subsets. Here, the spectra from instrument 1 were used. And the sample data of $Y$, $Z_1$, and $Z_2$ were standarized for simplicity. We trained PLFS models on training subset and evaluated the performances on test subset.
Mean squared prediction error (MSPE) was used to compare the predictive efficiency.
\[
MSPE = \frac{1}{n_{test}}\sum_{i=1}^{n_{test}}\big(Y_{i} - \wh{\mu}_{i}\big)^2,
\]
where $n_{test}$ is the size of test set.
As for candidate models, we considered $\{Z_1, Z_2\}$ for parametric components, and $\{\xi_1,\xi_2\}$, $\{\xi_1,\xi_2,\xi_3\}$, $\{\xi_1,\ldots,\xi_4\}$ as three kinds of candidate sets for FPC scores. Every candidate model comprised at least one component of parametric and nonparametric parts, respectively. Therefore, $M = 9, 21, 45$ corresponding to three settings.

\begin{table}
	\renewcommand\arraystretch{1.5}
	\caption{MSPE for {\it NIR shootout 2002} data.}
	\label{tab:shootout}
	\centering
	\begin{tabular}{lrrrrrr}
		\hline
		& $AIC$& $BIC$& $Equal$& $SAIC$& $SBIC$& $MMA$\\
		\hline
		$M=9$ & \textbf{0.7702}& 0.7957& 0.7825& 0.7796& 0.7933&			 0.7762 \\
		$M=21$& 		  0.7301 & 0.7301& 0.7379& 0.7203& 0.7494& \textbf{0.7177}\\
		$M=45$& 		  0.6956 & 0.7267& 0.7349& 0.6908& 0.7160& \textbf{0.6906}\\
		\hline
	\end{tabular}
\end{table}

Table~\ref{tab:shootout} presents the MSPE results of different procedures for {\it NIR shootout 2002} data set. We can observe from the table that for this data set, AIC generally yields smaller MSPEs than BIC does. Comparison between SAIC and SBIC also shows a similar pattern. Second, the relatively small MSPEs produced by MMA estimator indicates that MMA performs the best among these four model averaging methods. It is shown that equally weighting yields relatively large MSPEs under all settings, which suggests naive equally weighting precedure suffer severely in practice. Moreover, Table~\ref{tab:shootout} shows that MMA has an advantage over AIC in $M = 21, 45$ settings. Hence, MMA is likely to handle it well when encountered with diversity of candidate models.

\subsection{Equine articular cartilage data set}
This data set\footnote{Sarin, J. K. et~al. (2019). Dataset on equine cartilage near infrared spectra, composition, and functional properties. {\it figshare. Collection.} https://doi.org/10.6084/m9.figshare.c.4423139.v2} contains NIR spectroscopy measurements (functional predictor $X$) within the spectrum region of 700--1050 nm from 869 different locations across the articular surfaces of five equine fetlock joints, paired with comprehensive reference measurements from biomechanics, chemical composition and internal structure of the tissue, such as, instantaneous moduli ($Y$), collagen contents ($Z_1$), proteoglycan contents ($Z_2$), cartilage thickness ($Z_3$), and calcified layer thickness ($Z_4$), etc. More details are available \cite{sarinetal2019}. 
The sample data with size 530 were retained after removing all incomplete records. To evaluate the performance of our proposed procedure, we randomly selected 80\% of records as training set and constructed the test set using remaining records. Furthermore, we standardized sample data of $Z_1,\ldots, Z_4$, and performed successively logarithmic transformation and centralization on data of $Y$ to facilitate computation.
Two nested kinds of candidate models were considered: one included $\{Z_1, Z_2, Z_3\}+\{\xi_1, \ldots, \xi_4\}$, the other comprised
$\{Z_1, \ldots, Z_4\}+\{\xi_1, \ldots, \xi_4\}$. Each candidate model contained respectively at least one component of parametric and nonparametric parts, resulting in $M = 12, 16$. We conducted $D = 200$ runs. For each repetition, we still evaluated MSPE.
\[
MSPE^{(d)} = \frac{1}{n_{test}}\sum_{i=1}^{n_{test}}\big(Y_{i}^{(d)}-\wh{\mu}_{i}^{(d)}\big)^2, \quad d=1,\ldots,D,
\]
where $n_{test}$ is the size of test set, $Y_{i}^{(d)}$ is the $i$-th response of the $d$-th test set, and $\wh{\mu}_{i}^{(d)}$ is the prediction for $Y_{i}^{(d)}$. The average MSPEs with their standard error across $D$ repetitions were compared.

\begin{table}
	\renewcommand\arraystretch{1.5}
	\caption{The average MSPE with standard error on test set for Equine articular cartilage data.}
	\label{tab:equine}
	\centering
	\begin{tabular}{lrrrrrr}
		\hline
		& $AIC$& $BIC$& $Equal$& $SAIC$& $SBIC$& $MMA$\\
		\hline
		$M=12$	& 0.37908& 0.37908& 0.44997& 0.37908& 0.37908& \textbf{0.37856}\\
		& (0.00057)& (0.00057)& (0.00050)& (0.00057)& (0.00058)& (0.00055)\\
		$M=16$	& 0.63542& 0.63694& 0.69663& 0.63516& 0.63565& \textbf{0.63494}\\
		& (0.00001)& (0.00002)& (0.00000)& (0.00000)& (0.00002)& (0.00000)\\
		\hline
	\end{tabular}
\end{table}

Table~\ref{tab:equine} illustrates the average MSPEs with their standard errors across replications on test set.
First, it can be seen that MMA delivers the smallest results in terms of MSPE, which demonstrates the superiority of MMA to other model averaging and selection estimators and verifies the better prediction accuracy of MMA.
SAIC performs the second with larger MSPEs, slightly inferior to MMA. And AIC shares no smaller  average MSPEs than SAIC and MMA.
Second, BIC and SBIC do not perform well on this data set. The average MSPEs of BIC and SBIC are of larger scales compared to other methods.
Equally weighting method performs the worst with much larger mean MSPEs, which again emphasizes prudent use of it in practice.
The average MSPEs of model averaging estimators are always smaller than that of their model selection counterpart, which indicates that model averaging is a satisfactory alternative to model selection when prediction effect is of primary interest.
To sum up, these results show that the proposed MMA procedure is able to effectively deliver competitive outcomes.

\section{Conclusion and discussion}\label{sec6}

We presented a Mallows-type model averaging approach for PLFS model in which a scalar response depends both on scalar covariates and a functional predictor.
We verified the asymptotic optimality of MMA estimator when the function predictor is densely measured with error.
Additionally, a finite sample simulation was used to demonstrate the performance of the proposed estimator is either superior or comparable to that of classic competing model selection and averaging methods. Also, real data analysis manifested that the proposed estimator generally facilitated modification of the prediction results and reduced the possibility of producing poor outcomes when using a single model.

Many aspects deserve future research. In practice, if lots of variables are available, it would be reasonable to derive a suitable model averaging estimation for high-dimensional regression problems.
Moreover, there is room for studying the situation in which functional data are sparsely or irregularly observed, similar to cases in longitudinal studies. Besides, if more than one functional predictor were to exist, it would be interesting to determine how to effectively and efficiently conduct model averaging.
Finally, our asymptotic optimality is derived on the base of that all candidate models are misspecified. A consistent estimator is more desired if the correct model exists in our candidate set. Therefore, considering a consistent model averaging approach would be an avenue for future research.

%%%%%%%%%%%%%%%%%%%%%%%%%%%%%%%%%%%%%%%%%%%%%%%%%%%%%%%%%%%%%%%%%%55
%%%%%%%%%%%%%%%%%%%%%%%%%%%%%%%%%%%%%%%%%%%%%%%%%%%%%
\appendix

\section{Appendix}\label{app}

%%%%%%%%%%%%%%%%%%%%%%%%%%%%%%%%%%%%%%%%%%%%%%%%%%%5
%%%%%%%%%%%%%%%%%%%%%%%%%%%%%%%%%%%%%%5
\subsection{Some lemmas}

The estimation error of the transformed FPC score is of order $O_p(n^{-1/2}k)$ as shown in \cite{wongetal2019}, and we list the result here while omitting the detailed proof.
%%%%%%%%%%%%%%%%%%%%%%%%%%%%%
\begin{lemma}\label{lem1}
	Suppose the transformation function $\Phi(\cdot)$ has bounded derivative. Under Conditions (C1)--(C2), there is a constant $C>0$ such that $\mathbb{E}(\wh{\xi}_{ik} - \xi_{ik})^2 \leq Ck^2/n$ uniformly for $k\leq J_n$, where $J_n = \lfloor(2C_{\lambda}O_p(1))^{-1/(1+\alpha)}n^{1/(2+2\alpha)}\rfloor$.
\end{lemma}

%%%%%%%%%%%%%%%%%%%%%%%%%%%%%
\begin{lemma}\label{lem2}
	Under Conditions (C1)--(C4) and (C8), we have $\lambda_{\max}(\bK_{(m)}) = O(1)$, $\lambda_{\max}(\wh{\bK}_{(m)}) = O_p(1)$, and $\lambda_{\max}(\bP_{(m)}) = O(1)$, $\lambda_{\max}(\wh{\bP}_{(m)}) = O_p(1)$, for $m=1,\ldots, M$.
\end{lemma}
%%---------------------
\begin{proof}
	For any square matrices $\bM_1$ and $\bM_2$ (see \cite{li1987}), we have
	\begin{equation}\label{lamleq}
		\begin{aligned}
			\lambda_{\max}(\bM_1\bM_2) &\leq \lambda_{\max}(\bM_1)\lambda_{\max}(\bM_2) \\
			\mbox{and}\: \lambda_{\max}(\bM_1 + \bM_2) &\leq \lambda_{\max}(\bM_1) + \lambda_{\max}(\bM_2).
		\end{aligned}
	\end{equation}
	These two inequalities will be frequently used in the following proofs.
	
	By an inequality of Reisz (see \cite{speckman1988}), we obtain that
	\[
	\lambda_{\max}^2(\bK_{(m)})\leq \max_i \sum_{j=1}^n |\bK_{(m),ij}|\cdot \max_j \sum_{i=1}^n |\bK_{(m),ij}|,
	\]
	which implies that $\lambda_{\max}(\bK_{(m)})=O(1)$. Hence,
	\[
	\begin{aligned}
		\lambda_{\max}(\bP_{(m)}) &= \lambda_{\max}(\wt{\bP}_{(m)})\big(1+\lambda_{\max}(\bK_{(m)})\big) + \lambda_{\max}(\bK_{(m)}) \\
		&= \big(1+\lambda_{\max}(\bK_{(m)})\big) + \lambda_{\max}(\bK_{(m)}) = O(1).
	\end{aligned}
	\]
	
	From Lemma \ref{lem1}, we obtain that
	\[
	\begin{aligned}
		\wh{\xi}_{ik} - \xi_{ik} &= O_p(n^{-1/2}k), \\
		\wh{\xi}_{ik} - \wh{\xi}_{jk} &= \xi_{ik} - \xi_{jk} + O_p(n^{-1/2}k), \quad k\leq J_n.
	\end{aligned}
	\]
	Applying Taylor series expansion and condition (C4),
	\[
	\begin{aligned}
		& \wh{\bK}_{(m),ij} = \frac{\mathcal{K}(\wh{\bxi}_{(m),i} - \wh{\bxi}_{(m),j})}{\sum_{j'=1}^{n}\mathcal{K}(\wh{\bxi}_{(m),i} - \wh{\bxi}_{(m),j'})} \\
		& = \Big\{ \mathcal{K}(\bxi_{(m),i} - \bxi_{(m),j}) + \sum_{l=1}^{q_m}k'(\xi_{il}-\xi_{jl})\prod_{m\neq l}k(\xi_{im}-\xi_{jm})(\wh{\xi}_{il}-\xi_{il}+\xi_{jl}-\wh{\xi}_{jl}) \\
		& \quad +o_p(n^{-\frac{1}{2}}q_m)\Big\} / \Big\{\sum_{j'=1}^n [\mathcal{K}(\bxi_{(m),i} - \bxi_{(m),j'}) + O_p(n^{-\frac{3}{2}}q_m)]\Big\} \\
		& =\frac{\mathcal{K}(\bxi_{(m),i} - \bxi_{(m),j})}{\sum_{j'=1}^n\mathcal{K}(\bxi_{(m),i} - \bxi_{(m),j'})} + \Big\{\sum_{l=1}^{q_m}k'(\xi_{il}-\xi_{jl})\prod_{m\neq l}k(\xi_{im}-\xi_{jm})(\wh{\xi}_{il}-\xi_{il}+\xi_{jl}-\wh{\xi}_{jl})\Big\} \\
		& \quad / \Big\{\sum_{j'=1}^n \mathcal{K}(\bxi_{(m),i}-\bxi_{(m),j'})\Big\} + o_p(n^{-\frac{1}{2}}q_m) \\
		& = \bK_{(m),ij} + O_p(n^{-\frac{3}{2}}q_m), \qquad q_m\leq J_n,
	\end{aligned}
	\]
	i.e., $\wh{\bK}_{(m),ij} = \bK_{(m),ij} + O_p(n^{-\frac{3}{2}}q_m)$.
	Note that $q_m$ is no larger than $J_n$ and it is common for kernel smoothing to restrict the dimension of $\bxi$ to handle the curse of dimensionality. By conditions (C4) and (C8), we can show that
	\[
	\begin{aligned}
		\max_i \sum_{j=1}^n |\wh{\bK}_{(m),ij}| &= \max_i \sum_{j=1}^n |\bK_{(m),ij}| + O_p(n^{-\frac{1}{2}}q_m) = O_p(1), \\
		\max_j \sum_{i=1}^n |\wh{\bK}_{(m),ij}| &= \max_j \sum_{i=1}^n |\bK_{(m),ij}| + O_p(n^{-\frac{1}{2}}q_m) = O_p(1),
	\end{aligned}
	\]
	uniformly for $m = 1, \ldots, M$. Similarly,
	\[
	\lambda_{\max}^2(\wh{\bK}_{(m)})\leq \max_i \sum_{j=1}^n |\wh{\bK}_{(m),ij}| \max_j \sum_{i=1}^n |\wh{\bK}_{(m),ij}| = O_p(1),
	\]
	\[
	\begin{aligned}
		\lambda_{\max}(\wh{\bP}_{(m)}) &= \lambda_{\max}(\overline{\bP}_{(m)})\big(1+\lambda_{\max}(\wh{\bK}_{(m)})\big) + \lambda_{\max}(\wh{\bK}_{(m)}) \\
		&= \big(1+\lambda_{\max}(\wh{\bK}_{(m)})\big) + \lambda_{\max}(\wh{\bK}_{(m)}) = O_p(1).
	\end{aligned}
	\]
	In addition,
	\[
	\max_i\sum_{j=1}^n|\wh{\bK}_{(m),ij}-\bK_{(m),ij}|=O_p(n^{-\frac{1}{2}}q_m),
	\]
	
	\[
	\max_j\sum_{i=1}^n|\wh{\bK}_{(m),ij}-\bK_{(m),ij}|=O_p(n^{-\frac{1}{2}}q_m),
	\]
	
	\[
	\lambda_{\max}^2(\wh{\bK}_{(m)}-\bK_{(m)}) \leq \max_i\sum_{j=1}^n|\wh{\bK}_{(m),ij}-\bK_{(m),ij}| \max_j\sum_{i=1}^n|\wh{\bK}_{(m),ij}-\bK_{(m),ij}|,
	\]
	which leads to
	\begin{equation}\label{lamkk}
		\lambda_{\max}(\wh{\bK}_{(m)}-\bK_{(m)}) = O_p(n^{-\frac{1}{2}}q_m).
	\end{equation}
\end{proof}

%%%%%%%%%%%%%%%%%%%%%%%%%%%%%%
\begin{lemma}\label{lem3}
	Under Conditions (C1)--(C4) and (C7)--(C8), we have
	\[
	\lambda_{\max}\big(\bP_{(m)}-\wh{\bP}_{(m)}\big) = O_p(n^{-\frac{1}{2}}q_m + n^{-1}q_m^2),
	\]
	for $m=1, \ldots, M$.
\end{lemma}

\begin{proof}
	For $m=1, \ldots, M$, we have
	\begin{equation}\label{c1}
		\begin{aligned}
			\bP_{(m)} - \wh{\bP}_{(m)} &= (\wt{\bP}_{(m)} - \overline{\bP}_{(m)}) + (\bK_{(m)} - \wh{\bK}_{(m)}) + \wt{\bP}_{(m)}(\wh{\bK}_{(m)}-\bK_{(m)}) \\
			& + (\overline{\bP}_{(m)} - \wt{\bP}_{(m)})\bK_{(m)} + (\wt{\bP}_{(m)} - \overline{\bP}_{(m)})(\bK_{m} - \wh{\bK}_{(m)}).
		\end{aligned}
	\end{equation}
	
	Recalling (\ref{lamleq}), it suffices to determine the order of $\lambda_{\max}(\wt{\bP}_{(m)} - \overline{\bP}_{(m)})$ and $\lambda_{\max}(\bK_{(m)} - \wh{\bK}_{(m)})$. We have already quantified $\lambda_{\max}(\bK_{(m)} - \wh{\bK}_{(m)})$ in Eq. (\ref{lamkk}). Remind that $\wt{\bP}_{(m)}$ and $\overline{\bP}_{(m)}$ are projection matrices related to $\wt{\bZ}_{(m)} = (\bI - \bK_{(m)})\bZ_{(m)}$ and $\wh{\bZ}_{(m)} = (\bI - \wh{\bK}_{(m)})\bZ_{(m)}$, i.e. $\wt{\bP}_{(m)} = \wt{\bZ}_{(m)}(\wt{\bZ}^T_{(m)}\wt{\bZ}_{(m)})^{-1}\wt{\bZ}^T_{(m)}$ and $\overline{\bP}_{(m)} = \wh{\bZ}_{(m)}(\wh{\bZ}^T_{(m)}\wh{\bZ}_{(m)})^{-1}\wh{\bZ}^T_{(m)}$. We simplify the notations as $\wt{\bP}_{(m)} = \bH_{(\bI-\bK_{(m)})\bZ_{(m)}}$ and $\overline{\bP}_{(m)} = \bH_{(\bI-\wh{\bK}_{(m)})\bZ_{(m)}}$. $\bH_{A}$ represents the projection operator generated from $A$. Suppose the eigen decompositions of $\bK_{(m)}$ and $\wh{\bK}_{(m)}$ are $\bK_{(m)} = V\Lambda V^T$ and $\wh{\bK}_{(m)} = \wh{V}\wh{\Lambda}\wh{V}^T$, respectively. $\Lambda, \wh{\Lambda}\in \mathbb{R}^{n\times n}$ are diagonal matrices of eigenvalues, $V, \wh{V}\in \mathbb{R}^{n\times n}$ denote the corresponding eigenvector matrices. It is obvious that the deviation between $\wt{\bP}_{(m)}$ and $\overline{\bP}_{(m)}$ derives from the difference between $(\bI-\bK_{(m)})$ and $(\bI-\wh{\bK}_{(m)})$, so it holds that
	\begin{equation}\label{lampp0}
		\begin{aligned}
			& \|\wt{\bP}_{(m)} - \overline{\bP}_{(m)}\| \leq \|\bH_{\bI-\bK_{(m)}} - \bH_{\bI-\wh{\bK}_{(m)}}\| \\
			& = \|\bH_{\bK_{(m)}}-\bH_{\wh{\bK}_{(m)}}\| \\
			& \leq \|\sqrt{2}\sin \Theta(V,\wh{V})\| \\
			& \leq C\lambda_{\min}(\bK_{(m)})\|\bK_{(m)}-\wh{\bK}_{(m)}\| \\
			& = O_p(n^{-\frac{1}{2}}q_m),
		\end{aligned}
	\end{equation}
	where $C$ is a constant. The second and third inequalities follow from \cite{yuetal2015}.
	
	Finally according to Eq. (\ref{c1}), combining (\ref{lamleq}), Eqs. (\ref{lamkk}) and (\ref{lampp0}), we have
	\[
	\lambda_{\max}(\bP_{(m)}-\wh{\bP}_{(m)}) = O_p(n^{-\frac{1}{2}}q_m + n^{-1}q_m^2),
	\]
	which completes the proof.
\end{proof}

%%%%%%%%%%%%%%%%%%%%%%%%%%%%%%%%%%%%%%%%%%%%%%%%%%%%%%
%%%%%%%%%%%%%%%%%%%%%%%%%%%%%%%%%%%%%%%%
\subsection{Proof of Theorem 1}

%\begin{theorem}\label{th1}
%Under Conditions (C1)--(C9), it holds that
%\begin{equation}\label{A1}
%\frac{L_n(\wt{\bomega})}{\inf_{\bomega\in \mathcal{H}_n}L_n(\bomega)}\to 1
%\end{equation}
%in probability as $n\to \infty$.
%\end{theorem}
%---------------------
\begin{proof}
	
	Firstly, it follows from Lemma \ref{lem2} that
	\begin{equation}\label{lamwh}
		\begin{aligned}
			\sup_{\bomega} \lambda_{\max}\big(\wh{\bP}(\bomega)\big) & = \sup_{\bomega} \lambda_{\max}\Big(\sum_{m=1}^M \omega_m \wh{\bP}_{(m)}\Big) \leq \sup_{\bomega} \sum_{m=1}^M \omega_m \lambda_{\max}(\wh{\bP}_{(m)}) \\
			&\leq \max_{1\leq m\leq M} \lambda_{\max}(\wh{\bP}_{(m)}) = O_p(1),
		\end{aligned}
	\end{equation}
	and similarly,
	\begin{equation}\label{lamw}
		\sup_{\bomega} \lambda_{\max}\big(\bP(\bomega)\big) = O_p(1).
	\end{equation}

	Let $\wh{\bA}(\bomega) = \bI-\wh{\bP}(\bomega)$ and $\bA(\bomega) = \bI-\bP(\bomega)$.
	From the definition of $L_n(\bomega)$, $\wh{C}_n(\bomega)$ and $R_n(\bomega)$, we have
	\[
	\begin{aligned}
		\wh{C}_n(\bomega) &= L_n(\bomega) + \|\varepsilon\|^2 - 2\mu^T(\bP(\bomega)-\wh{\bP}(\bomega))\mu -2\varepsilon^T(\bP(\bomega)-\wh{\bP}(\bomega))\varepsilon - 2\varepsilon^T(\bP(\bomega)-\wh{\bP}(\bomega))\mu \\
		&  - 2\varepsilon^T(\bP(\bomega)-\wh{\bP}(\bomega))^T\mu - 2\varepsilon^T(\wh{\bP}(\bomega)-\bP(\bomega))^T\bP(\bomega)\mu + 2\varepsilon^T\bP^T(\bomega)(\wh{\bP}(\bomega)-\bP(\bomega))\mu \\
		& + \mu^T(\bP(\bomega)+\wh{\bP}(\bomega))^T(\bP(\bomega)-\wh{\bP}(\bomega))\mu + \varepsilon^T(\bP(\bomega)+\wh{\bP}(\bomega))^T(\bP(\bomega)-\wh{\bP}(\bomega))\varepsilon \\
		& + \varepsilon^T(\bP(\bomega)+\wh{\bP}(\bomega))^T(\bP(\bomega)-\wh{\bP}(\bomega))\mu + 2\varepsilon^T\bA(\bomega)\mu \\
		& - 2[\varepsilon^T\bP(\bomega)\varepsilon - tr(\bP(\bomega)\bOmega)] - 2[tr(\bP(\bomega)\bOmega)- tr(\wh{\bP}(\bomega)\bOmega)],
	\end{aligned}
	\]
	and
	\[
	L_n(\bomega) - R_n(\bomega) = \varepsilon^T\bP^{T}(\bomega)\bP(\bomega)\varepsilon - tr(\bP^T(\bomega)\bP(\bomega)\bOmega) - 2\varepsilon^T\bP^T(\bomega)\bA(\bomega)\mu.
	\]
	So similar to the proof of Theorem 2.1 of \cite{li1987}, in order to prove Eq.(\ref{A1}), we need only to verify that
	\begin{equation}\label{A2}
		\sup_{\bomega} \frac{|\varepsilon^T\bA(\bomega)\mu|}{R_n(\bomega)} = o_p(1),
	\end{equation}
	\begin{equation}\label{A3}
		\sup_{\bomega} \frac{|\varepsilon^T\bP(\bomega)\varepsilon - tr(\bP(\bomega)\bOmega)|}{R_n(\bomega)} = o_p(1),
	\end{equation}
	\begin{equation}\label{A4}
		\sup_{\bomega}
		\frac{|\varepsilon^T\bP^{T}(\bomega)\bP(\bomega)\varepsilon - tr(\bP^{T}(\bomega)\bP(\bomega)\bOmega)|}{R_n(\bomega)} = o_p(1),
	\end{equation}
	\begin{equation}\label{A5}
		\sup_{\bomega}
		\frac{|\mu^T\bA^{T}(\bomega)\bP(\bomega)\varepsilon|}{R_n(\bomega)} = o_p(1),
	\end{equation}
	
	\begin{equation}\label{A6}
		\sup_{\bomega} \frac{|\varepsilon^T\big(\bP(\bomega)-\wh{\bP}(\bomega)\big)^T\mu|}{R_n(\bomega)} = o_p(1),
	\end{equation}
	\begin{equation}\label{A7}
		\sup_{\bomega} \frac{|\varepsilon^T\big(\bP(\bomega)-\wh{\bP}(\bomega)\big)\mu|}{R_n(\bomega)} = o_p(1),
	\end{equation}
	\begin{equation}\label{A8}
		\sup_{\bomega}
		\frac{|\varepsilon^T\big(\wh{\bP}(\bomega)-\bP(\bomega)\big)^T\bP(\bomega)\mu|}{R_n(\bomega)} = o_p(1),
	\end{equation}
	\begin{equation}\label{A9}
		\sup_{\bomega}
		\frac{|\varepsilon^T\bP^{T}(\bomega)\big(\wh{\bP}(\bomega)-\bP(\bomega)\big)\mu|}{R_n(\bomega)} = o_p(1),
	\end{equation}
	
	\begin{equation}\label{A10}
		\sup_{\bomega}
		\frac{|\varepsilon^T\big(\bP(\bomega)+\wh{\bP}(\bomega)\big)^T\big(\bP(\bomega)-\wh{\bP}(\bomega)\big)\mu|}{R_n(\bomega)} = o_p(1),
	\end{equation}
	\begin{equation}\label{A11}
		\sup_{\bomega}
		\frac{|\mu^T\big(\bP(\bomega)-\wh{\bP}(\bomega)\big)\mu|}{R_n(\bomega)} = o_p(1),
	\end{equation}
	\begin{equation}\label{A12}
		\sup_{\bomega}
		\frac{|\mu^T\big(\bP(\bomega)+\wh{\bP}(\bomega)\big)^T\big(\bP(\bomega)-\wh{\bP}(\bomega)\big)\mu|}{R_n(\bomega)} = o_p(1),
	\end{equation}
	\begin{equation}\label{A13}
		\sup_{\bomega} \frac{|\varepsilon^T\big(\bP(\bomega)-\wh{\bP}(\bomega)\big)\varepsilon|}{R_n(\bomega)} = o_p(1),
	\end{equation}
	\begin{equation}\label{A14}
		\sup_{\bomega}
		\frac{|\varepsilon^T\big(\bP(\bomega)+\wh{\bP}(\bomega)\big)^T\big(\bP(\bomega)-\wh{\bP}(\bomega)\big)\varepsilon|}{R_n(\bomega)} = o_p(1),
	\end{equation}
	
	and
	
	\begin{equation}\label{A15}
		\sup_{\bomega} \frac{|tr(\bP(\bomega)\bOmega)-tr(\wh{\bP}(\bomega)\bOmega)|}{R_n(\bomega)} = o_p(1).
	\end{equation}
	Note that Eqs. (\ref{A2})--(\ref{A5}) do not include any $\wh{\cdot}$ terms. From Eq. (\ref{lamw}) and conditions (C5)--(C6), Eqs. (\ref{A2})--(\ref{A5}) can be shown by using the same steps as in the proof of Theorem 1 of \cite{wanetal2010}.
	
	For proving Eq. (\ref{A14}), by (\ref{lamleq}), it is seen that
	\[
	\begin{aligned}
		& \sup_{\bomega}
		\frac{|\varepsilon^T\big(\bP(\bomega)+\wh{\bP}(\bomega)\big)^T\big(\bP(\bomega)-\wh{\bP}(\bomega)\big)\varepsilon|}{R_n(\bomega)}\\ & \leq \eta_n^{-1}\frac{1}{2}\sup_{\bomega}\left|\varepsilon^T\Big[\big(\bP(\bomega)+\wh{\bP}(\bomega)\big)^T\big(\bP(\bomega)-\wh{\bP}(\bomega)\big) + \big(\bP(\bomega)-\wh{\bP}(\bomega)\big)^T\big(\bP(\bomega)+\wh{\bP}(\bomega)\big)\Big]\varepsilon\right| \\
		& \leq \eta_n^{-1}\frac{1}{2}\|\varepsilon\|^2\cdot \sup_{\bomega}\lambda_{\max}\Big[\big(\bP(\bomega)+\wh{\bP}(\bomega)\big)^T\big(\bP(\bomega)-\wh{\bP}(\bomega)\big) + \big(\bP(\bomega)-\wh{\bP}(\bomega)\big)^T\big(\bP(\bomega)+\wh{\bP}(\bomega)\big)\Big] \\
		& \leq \eta_n^{-1}\|\varepsilon\|^2\cdot \sup_{\bomega}\lambda_{\max}\big(\bP(\bomega)+\wh{\bP}(\bomega)\big)\lambda_{\max}\big(\bP(\bomega)-\wh{\bP}(\bomega)\big) \\
		& \leq \eta_n^{-1}\|\varepsilon\|^2\cdot \sup_{\bomega}\big[\lambda_{\max}(\bP(\bomega))+\lambda_{\max}(\wh{\bP}(\bomega))\big]\cdot \sum_{m=1}^M \omega_m \lambda_{\max}(\bP_{(m)}-\wh{\bP}_{(m)}) \\
		& \leq n\eta_n^{-1} \cdot \frac{\|\varepsilon\|^2}{n}\cdot \sup_{\bomega}\big[\lambda_{\max}(\bP(\bomega))+\lambda_{\max}(\wh{\bP}(\bomega))\big]\cdot \max_{1\leq m\leq M} \lambda_{\max}(\bP_{(m)}-\wh{\bP}_{(m)}) \\
		& = o_p(1),
	\end{aligned}
	\]
	where the last step is from Eqs. (\ref{lamwh})--(\ref{lamw}), condition (C5) and Lemma \ref{lem3}. By Lemma \ref{lem3}, conditions (C5)--(C9), we can prove Eqs. (\ref{A10})--(\ref{A13}) in a similar way.

	For Eq. (\ref{A7}),
	\[
	\begin{aligned}
		& \sup_{\bomega} \frac{|\varepsilon^T\big(\bP(\bomega)-\wh{\bP}(\bomega)\big)\mu|}{R_n(\bomega)}\\
		& \leq \eta_n^{-1}\|\mu\|\cdot\sup_{\bomega}\|\big(\bP(\bomega)-\wh{\bP}(\bomega)\big)^T\varepsilon\|\\
		& \leq \eta_n^{-1}\cdot \|\mu\| \cdot\sup_{\bomega}\lambda_{\max}\big(\bP(\bomega)-\wh{\bP}(\bomega)\big)\cdot\|\varepsilon\| \\
		& \leq n\eta_n^{-1}\cdot\frac{\|\mu\|}{\sqrt{n}}\frac{\|\varepsilon\|}{\sqrt{n}}\cdot\max_{1\leq m\leq M} \lambda_{\max}\big(\bP_{(m)}-\wh{\bP}_{(m)}\big) \\
		& = o_p(1),
	\end{aligned}
	\]
	where the last step is from Lemma \ref{lem3}, conditions (C5), (C7)--(C9). Similarly, we can verify Eqs. (\ref{A6}), (\ref{A8})--(\ref{A9}) by Lemma \ref{lem3}, conditions (C5), (C7)--(C9) and Eqs. (\ref{lamwh})--(\ref{lamw}).

	Now we consider the last Eq. (\ref{A15}). Note that $\bOmega$ is a diagonal matrix,
	\[
	\begin{aligned}
		& \sup_{\bomega} \frac{|tr(\bP(\bomega)\bOmega)-tr(\wh{\bP}(\bomega))|}{R_n(\bomega)}\\ & = \sup_{\bomega}\frac{|tr[(\bP(\bomega)-\wh{\bP}(\bomega))\bOmega]|}{R_n(\bomega)}\\
		& \leq \eta_n^{-1}\sup_{\bomega}|tr(\bP(\bomega)-\wh{\bP}(\bomega))|\lambda_{\max}(\bOmega)\\
		& \leq n\eta_n^{-1}\max_{1\leq m\leq M}\lambda_{\max}\big(\bP_{(m)}-\wh{\bP}_{(m)}\big)\lambda_{\max}(\bOmega)\\
		& = o_p(1),
	\end{aligned}
	\]
	where the last step is from Lemma \ref{lem3}, conditions (C5), (C7)--(C8).
	This completes the proof of Theorem \ref{th1}.
\end{proof}

%%%%%%%%%%%%%%%%%%%%%%%%%%%%%%%%%%%%%%%%%%%%%%5
%%%%%%%%%%%%%%%%%%%%%%%%%%%%%%%%%%%%%5
\subsection{Proof of Theorem 2}

%\begin{theorem}\label{th2}
%Under Conditions (C1)--(C12), we have that
%\begin{equation}\label{B1}
%\frac{L_n(\wh{\bomega})}{\inf_{\bomega\in\mathcal{H}_n} L_n(\bomega)}\to 1
%\end{equation}
%in probability as $n\to \infty$.
%\end{theorem}
%----------------------------------
\begin{proof}
	Note that
	\[
	\wh{C}_n(\bomega)|_{\bOmega = \wh{\bOmega}} = \wh{C}_n(\bomega) + 2tr(\wh{\bP}(\bomega)\wh{\bOmega})-2tr(\wh{\bP}(\bomega)\bOmega).
	\]
	From the result of Theorem \ref{th1}, to prove Eq. (\ref{B1}), it suffices to prove that
	\begin{equation}\label{B2}
		\sup_{\bomega} \frac{|tr(\wh{\bP}(\bomega)\wh{\bOmega})-tr(\wh{\bP}(\bomega)\bOmega)|}{R_n(\bomega)} = o_p(1).
	\end{equation}
	Let $\bQ_{(m)} = diag(\rho_{11}^{(m)}, \ldots, \rho_{nn}^{(m)})$ and $\bQ(\bomega) = \sum_{m=1}^M \omega_m \bQ_{(m)}$. To prove Eq. (\ref{B2}), we decompose the left-hand side of Eq. (\ref{B2}) into four parts as follows.
	\[
	\begin{aligned}
		& \sup_{\bomega} \frac{|tr(\wh{\bP}(\bomega)\wh{\bOmega})-tr(\wh{\bP}(\bomega)\bOmega)|}{R_n(\bomega)}\\
		& = \sup_{\bomega} \frac{|(Y-\wh{\bP}_{(M^*)}Y)^T\wh{\bQ}(\bomega)(Y-\wh{\bP}_{(M^*)}Y)-tr(\wh{\bQ}(\bomega)\bOmega)|}{R_n(\bomega)}\\
		& = \sup_{\bomega} \frac{|(\mu+\varepsilon)^T(\bI-\wh{\bP}_{(M^*)})^T\wh{\bQ}(\bomega)(\bI-\wh{\bP}_{(M^*)})(\mu+\varepsilon)-tr(\wh{\bQ}(\bomega)\bOmega)|}{R_n(\bomega)}\\
		& \leq \sup_{\bomega} \frac{|\mu^T(\bI-\wh{\bP}_{(M^*)})^T\wh{\bQ}(\bomega)(\bI-\wh{\bP}_{(M^*)})\mu|}{R_n(\bomega)} + \sup_{\bomega} \frac{2|\varepsilon^T(\bI-\wh{\bP}_{(M^*)})^T\wh{\bQ}(\bomega)(\bI-\wh{\bP}_{(M^*)})\mu|}{R_n(\bomega)}\\
		& \quad + \sup_{\bomega} \frac{|\varepsilon^T(\bI-\wh{\bP}_{(M^*)})^T\wh{\bQ}(\bomega)(\bI-\wh{\bP}_{(M^*)})\varepsilon|}{R_n(\bomega)} + \sup_{\bomega} \frac{|tr(\wh{\bQ}(\bomega)\bOmega)|}{R_n(\bomega)} \\
		& \equiv \Xi_1 + \Xi_2 + \Xi_3 + \Xi_4.
	\end{aligned}
	\]
	
	Now define $\rho = \max_{1\leq m\leq M}\max_{1\leq i\leq n} |\rho_{ii}^{(m)}|$. From conditions (C10)--(C11) and Lemma \ref{lem2}, we have $\max_{1\leq m\leq M}|tr(\wh{\bK}_{(m)})| = \max_{1\leq m\leq M}|tr(\bK_{(m)})| + O_p(n^{-\frac{1}{2}}\wt{q})$ and
	\begin{equation}\label{B3}
		\begin{aligned}
			\rho & \leq cn^{-1}\max_{1\leq m\leq M}|tr(\wh{\bP}_{(m)})| \\
			& \leq cn^{-1}\max_{1\leq m \leq M}|tr(\overline{\bP}_{(m)})| + cn^{-1}\max_{1\leq m\leq M}|tr(\overline{\bP}_{(m)}\wh{\bK}_{(m)})| + cn^{-1}\max_{1\leq m\leq M}|tr(\wh{\bK}_{(m)})| \\
			& \leq cn^{-1}\max_{1\leq m\leq M}rank(\overline{\bP}_{(m)}) + cn^{-1}\frac{1}{2}\max_{1\leq m\leq M}\big[\lambda_{\max}\big(\overline{\bP}_{(m)}\wh{\bK}_{(m)} + \wh{\bK}_{(m)}^T\overline{\bP}_{(m)}\big) \\
			& \qquad \cdot rank\big(\overline{\bP}_{(m)}\wh{\bK}_{(m)} + \wh{\bK}_{(m)}^T\overline{\bP}_{(m)}\big)\big] + cn^{-1}\max_{1\leq m\leq M}|tr(\wh{\bK}_{(m)})| \\
			& \leq cn^{-1}\wt{p} + cn^{-1}\cdot2\wt{p}\cdot\lambda_{\max}(\overline{\bP}_{(m)})\lambda_{\max}(\wh{\bK}_{(m)}) + cn^{-1}\max_{1\leq m\leq M}|tr(\wh{\bK}_{(m)})| \\
			& = cn^{-1}\wt{p} + cn^{-1}\wt{p}\cdot O_p(1) + cn^{-1}\cdot O_p(h^{-\wt{q}}+n^{-\frac{1}{2}}\wt{q}) \\
			& = O_p(n^{-1}\wt{p} + n^{-1}h^{-\wt{q}} + n^{-\frac{3}{2}}\wt{q}).
		\end{aligned}
	\end{equation}
	
	It follows from Lemma \ref{lem2}, conditions (C9)--(C10) and Eqs. (\ref{lamwh}) and (\ref{B3}) that
	\[
	\begin{aligned}
		\Xi_1 & \leq \eta_n^{-1}\sup_{\bomega}\lambda_{\max}(\wh{Q}(\bomega))\cdot\|(\bI-\wh{\bP}_{(M^*)})\mu\|^2\\
		& \leq \eta_n^{-1}\rho\cdot\|(\bI-\wh{\bP}_{(M^*)})\mu\|^2\\
		& \leq \eta_n^{-1}\rho\cdot\big[1+\lambda_{\max}(\wh{\bP}_{(M^*)})\big]^2\cdot\|\mu\|^2\\
		& = \eta_n^{-1}\cdot O_p(n^{-1}\wt{p} + n^{-1}h^{-\wt{q}} + n^{-\frac{3}{2}}\wt{q})\cdot O_p(1)\cdot O_p(n) \\
		& = O_p(\eta_n^{-1}\wt{p} + \eta_n^{-1}h^{-\wt{q}} + n^{-\frac{1}{2}}\eta_n^{-1}\wt{q}).
	\end{aligned}
	\]
	
	Using Lemma \ref{lem2}, conditions (C5), (C9)--(C10) and Eqs. (\ref{lamwh}) and (\ref{B3}), we obtain that
	\[
	\begin{aligned}
		\Xi_2 & \leq 2\eta_n^{-1}\cdot \|(\bI-\wh{\bP}_{(M^*)})\mu\|\cdot \sup_{\bomega}\|\wh{\bQ}(\bomega)(\bI-\wh{\bP}_{(M^*)})\varepsilon\|\\
		& \leq 2\eta_n^{-1}\cdot \|(\bI-\wh{\bP}_{(M^*)})\mu\|\cdot \sup_{\bomega}\lambda_{\max}(\wh{\bQ}(\bomega))\cdot \|(\bI-\wh{\bP}_{(M^*)})\varepsilon\| \\
		& \leq 2\eta_n^{-1}\cdot \|(\bI-\wh{\bP}_{(M^*)})\mu\|\cdot \rho\cdot \|(\bI-\wh{\bP}_{(M^*)})\varepsilon\| \\
		& \leq 2\eta_n^{-1}\cdot \big(1+\lambda_{\max}(\wh{\bP}_{(M^{*})})\big)\cdot \|\mu\|\cdot \rho\cdot \big(1+\lambda_{\max}(\wh{\bP}_{(M^{*})})\big)\cdot \|\varepsilon\| \\
		& = 2\eta_n^{-1}\cdot O_p(1)\cdot O_p(n^{\frac{1}{2}})\cdot O_p(n^{-1}\wt{p}+n^{-1}h^{-\wt{q}} + n^{-\frac{3}{2}}\wt{q})\cdot O_p(1)\cdot O_p(n^{\frac{1}{2}}) \\
		& = O_p(\eta_n^{-1}\wt{p} + \eta_n^{-1}h^{-\wt{q}} + n^{-\frac{1}{2}}\eta_{n}^{-1}\wt{q}).
	\end{aligned}
	\]
	
	Using Lemma \ref{lem2}, conditions (C5) and (C10), and Eqs. (\ref{lamwh}) and (\ref{B3}), we have
	\[
	\begin{aligned}
		\Xi_3 & \leq \eta_n^{-1}\cdot \sup_{\bomega}\lambda_{\max}(\wh{\bQ}(\bomega))\cdot \|(\bI-\wh{\bP}_{(M^*)})\varepsilon\|^2 \\
		& \leq \eta_n^{-1}\cdot \rho \big[1+\lambda_{\max}(\wh{\bP}_{(M^*)})\big]^2\cdot \|\varepsilon\|^2 \\
		& = \eta_n^{-1}\cdot O_p(n^{-1}\wt{p} + n^{-1}h^{-\wt{q}} + n^{-\frac{3}{2}}\wt{q})\cdot O_p(1)\cdot O_p(n) \\
		& = O_p(\eta_n^{-1}\wt{p} + \eta_n^{-1}h^{-\wt{q}}).
	\end{aligned}
	\]
	
	Using conditions (C5) and (C10), and Eqs. (\ref{lamwh}) and (\ref{B3}), we have
	\[
	\begin{aligned}
		\Xi_4 & \leq \eta_n^{-1}\cdot n\sup_{\bomega} \lambda_{\max}(\wh{\bQ}(\bomega))\cdot \lambda_{\max}(\bOmega) \\
		& \leq \eta_n^{-1}\cdot n\cdot \rho\cdot \lambda_{\max}(\bOmega) \\
		& = \eta_n^{-1}\cdot n\cdot O_p(n^{-1}\wt{p} + n^{-1}h^{-\wt{q}} + n^{-\frac{3}{2}}\wt{q})\cdot O_p(1) \\
		& = O_p(\eta_n^{-1}\wt{p} + \eta_n^{-1}h^{-\wt{q}} + n^{-\frac{1}{2}}\eta_n^{-1}\wt{q}).
	\end{aligned}
	\]
	
	Finally, it follows from conditions (C7)--(C8) and (C12) that $\Xi_1 = o_p(1)$, $\Xi_2 = o_p(1)$, $\Xi_3 = o_p(1)$ and $\Xi_4 = o_p(1)$. Therefore, we have verified Eq. (\ref{B2}) and this completes the proof.
	
\end{proof}

%\bibliographystyle{imsart-number}
%\bibliography{refs}

\begin{thebibliography}{99}
	
	\bibitem{akaike1973}
	Akaike, H. (1973). Maximum likelihood identification of Gaussian autoregressive moving average models. {\it Biometrika}, 60(2), 255-265.
	
	\bibitem{andrews1991}
	Andrews, D. W. (1991). Asymptotic optimality of generalized CL, cross-validation, and generalized cross-validation in regression with heteroskedastic errors. {\it Journal of Econometrics}, 47(2-3), 359-377.
	
	\bibitem{bairetal2006}
	Bair, E., Hastie, T., Paul, D., \& Tibshirani, R. (2006). Prediction by supervised principal components. {\it Journal of the American Statistical Association}, 101(473), 119-137.
	
	\bibitem{bucklandetal1997}
	Buckland, S. T., Burnham, K. P., \& Augustin, N. H. (1997). Model selection: an integral part of inference. {\it Biometrics}, 53, 603-618.
	
	\bibitem{caihall2006}
	Cai, T. T., \& Hall, P. (2006). Prediction in functional linear regression. {\it Annals of Statistics}, 34(5), 2159-2179.
	
	\bibitem{caiyuan2012}
	Cai, T. T., \& Yuan, M. (2012). Minimax and adaptive prediction for functional linear regression. {\it Journal of the American Statistical Association}, 107(499), 1201-1216.
	
	\bibitem{cardotetal1999}
	Cardot, H., Ferraty, F., \& Sarda, P. (1999). Functional linear model. {\it Statistics \& Probability Letters}, 45(1), 11-22.
	
	\bibitem{cardotetal2003}
	Cardot, H., Ferraty, F., \& Sarda, P. (2003). Spline estimators for the functional linear model. {\it Statistica Sinica}, 13, 571-591.
	
	\bibitem{chenghansen2015}
	Cheng, X., \& Hansen, B. E. (2015). Forecasting with factor-augmented regression: A frequentist model averaging approach. {\it Journal of Econometrics}, 186(2), 280-293.
	
	\bibitem{crainiceanuetal2009}
	Crainiceanu, C. M., Staicu, A. M., \& Di, C. Z. (2009). Generalized multilevel functional regression. {\it Journal of the American Statistical Association}, 104(488), 1550-1561.
	
	\bibitem{gaoetal2016}
	Gao, Y., Zhang, X., Wang, S., \& Zou, G. (2016). Model averaging based on leave-subject-out cross-validation. {\it Journal of Econometrics}, 192(1), 139-151.
	
	\bibitem{hadiling1998}
	Hadi, A. S., \& Ling, R. F. (1998). Some cautionary notes on the use of principal components regression. {\it American Statistician}, 52(1), 15-19.
	
	\bibitem{halletal2006}
	Hall, P., M\"{u}ller, H. G., \& Wang, J. L. (2006). Properties of principal component methods for functional and longitudinal data analysis. {\it Annals of Statistics}, 34(3), 1493-1517.
	
	\bibitem{hansen2007}
	Hansen, B. E. (2007). Least squares model averaging. {\it Econometrica}, 75(4), 1175-1189.
	
	\bibitem{hansenracine2012}
	Hansen, B. E., \& Racine, J. S. (2012). Jackknife model averaging. {\it Journal of Econometrics}, 167(1), 38-46.
	
	\bibitem{hjortclaeskens2003}
	Hjort, N. L., \& Claeskens, G. (2003). Frequentist model average estimators. {\it Journal of the American Statistical Association}, 98(464), 879-899.
	
	\bibitem{hoetingetal1999}
	Hoeting, J. A., Madigan, D., Raftery, A. E., \& Volinsky, C. T. (1999). Bayesian model averaging: a tutorial. {\it Statistical Science}, 14, 382-401.
	
	\bibitem{james2002}
	James, G. M. (2002). Generalized linear models with functional predictors. {\it Journal of the Royal Statistical Society: Series B (Statistical Methodology)}, 64(3), 411-432.
	
	\bibitem{kongetal2016}
	Kong, D., Xue, K., Yao, F., \& Zhang, H. H. (2016). Partially functional linear regression in high dimensions. {\it Biometrika}, 103(1), 147-159.
	
	\bibitem{li1987}
	Li, K. C. (1987). Asymptotic optimality for $ C_p, C_L $, cross-validation and generalized cross-validation: Discrete index set. {\it Annals of Statistics}, 15(3), 958-975.
	
	\bibitem{lietal2010}
	Li, Y., Wang, N., \& Carroll, R. J. (2010). Generalized functional linear models with semiparametric single-index interactions. {\it Journal of the American Statistical Association}, 105(490), 621-633.
	
	\bibitem{liangetal2011}
	Liang, H., Zou, G., Wan, A. T., \& Zhang, X. (2011). Optimal weight choice for frequentist model average estimators. {\it Journal of the American Statistical Association}, 106(495), 1053-1066.
	
	\bibitem{liuokui2013}
	Liu, Q., \& Okui, R. (2013). Heteroskedasticity-robust Cp model averaging. {\it Econometrics Journal}, 16, 463-472.
	
	\bibitem{morris2015}
	Morris, J. S. (2015). Functional regression. {\it Annual Review of Statistics and Its Application}, 2, 321-359.
	
	\bibitem{mullerstadtmuller2005}
	M\"{u}ller, H. G., \& Stadtm\"{u}ller, U. (2005). Generalized functional linear models. {\it Annals of Statistics}, 33(2), 774-805.
	
	\bibitem{mulleryao2008}
	M\"{u}ller, H. G., \& Yao, F. (2008). Functional additive models. {\it Journal of the American Statistical Association}, 103(484), 1534-1544.
	
	\bibitem{ramsaysilverman2005}
	Ramsay, J. O., \& Silverman, B. W. (2005), Functional Data Analysis (2nd ed.), New York: Springer.
	
	\bibitem{ricesilverman1991}
	Rice, J. A., \& Silverman, B. W. (1991). Estimating the mean and covariance structure nonparametrically when the data are curves. {\it Journal of the Royal Statistical Society: Series B (Methodological)}, 53(1), 233-243.
	
	\bibitem{sangetal2018}
	Sang, P., Lockhart, R. A., \& Cao, J. (2018). Sparse estimation for functional semiparametric additive models. {\it Journal of Multivariate Analysis}, 168, 105-118.
	
	\bibitem{sarinetal2019}
	Sarin, J. K., Torniainen, J., Prakash, M., Rieppo, L., Afara, I. O., \& T\"{o}yr\"{a}s, J. (2019). Dataset on equine cartilage near infrared spectra, composition, and functional properties. {\it Scientific data}, 6(1), 1-8.
	
	\bibitem{schwarz1978}
	Schwarz, G. (1978). Estimating the dimension of a model. {\it Annals of Statistics}, 6(2), 461-464.
	
	\bibitem{speckman1988}
	Speckman, P. (1988). Kernel smoothing in partial linear models. {\it Journal of the Royal Statistical Society: Series B (Methodological)}, 50(3), 413-436.
	
	\bibitem{wanetal2010}
	Wan, A. T., Zhang, X., \& Zou, G. (2010). Least squares model averaging by Mallows criterion. {\it Journal of Econometrics}, 156(2), 277-283.
	
	\bibitem{wongetal2019}
	Wong, R. K., Li, Y., \& Zhu, Z. (2019). Partially linear functional additive models for multivariate functional data. {\it Journal of the American Statistical Association}, 114(525), 406-418.
	
	\bibitem{yang2001}
	Yang, Y. (2001). Adaptive regression by mixing. {\it Journal of the American Statistical Association}, 96(454), 574-588.
	
	\bibitem{yang2003}
	Yang, Y. (2003). Regression with multiple candidate models: selecting or mixing?. {\it Statistica Sinica}, 13, 783-809.
	
	\bibitem{yaoetal2005a}
	Yao, F., M\"{u}ller, H. G., \& Wang, J. L. (2005a). Functional linear regression analysis for longitudinal data. {\it Annals of Statistics}, 33(6), 2873-2903.
	
	\bibitem{yaoetal2005b}
	Yao, F., M\"{u}ller, H. G., \& Wang, J. L. (2005b). Functional data analysis for sparse longitudinal data. {\it Journal of the American Statistical Association}, 100(470), 577-590.
	
	\bibitem{yuetal2015}
	Yu, Y., Wang, T., \& Samworth, R. J. (2015). A useful variant of the Davis–Kahan theorem for statisticians. {\it Biometrika}, 102(2), 315-323.
	
	\bibitem{zhangzou2020}
	Zhang, H., \& Zou, G. (2020). Cross-Validation Model Averaging for Generalized Functional Linear Model. {\it Econometrics}, 8(1), 7.
	
	\bibitem{zhangetal2018}
	Zhang, X., Chiou, J. M., \& Ma, Y. (2018). Functional prediction through averaging estimated functional linear regression models. {\it Biometrika}, 105(4), 945-962.
	
	\bibitem{zhangwang2019}
	Zhang, X., \& W. Wang. (2019). Optimal model averaging estimation for partially linear models. {\it Statistica Sinica}, 29, 693-718.
	
	\bibitem{zhangyu2018}
	Zhang, X., \& Yu, J. (2018). Spatial weights matrix selection and model averaging for spatial autoregressive models. {\it Journal of Econometrics}, 203(1), 1-18.
	
	\bibitem{zhangetal2014}
	Zhang, X., Zou, G., \& Liang, H. (2014). Model averaging and weight choice in linear mixed-effects models. {\it Biometrika}, 101(1), 205-218.
	
	\bibitem{zhuetal2007}
	Zhu, H., Vannucci, M., \& Cox, D. D. (2007). Functional data classification in cervical pre-cancer diagnosis-a bayesian variable selection model. {\it Proc. Jt Statist. Meet.}
	
	\bibitem{zhuetal2014}
	Zhu, H., Yao, F., \& Zhang, H. H. (2014). Structured functional additive regression in reproducing kernel Hilbert spaces. {\it Journal of the Royal Statistical Society: Series B (Statistical Methodology)}, 76(3), 581-603.
	
	\bibitem{zhuetal2019}
	Zhu, R., Wan, A. T., Zhang, X., \& Zou, G. (2019). A Mallows-type model averaging estimator for the varying-coefficient partially linear model. {\it Journal of the American Statistical Association}, 114(526), 882-892.
	
	\bibitem{zhuetal2018}
	Zhu, R., Zou G., \& Zhang X. (2018). Optimal model averaging estimation for partial functional linear models. {\it Journal of Systems Science and Mathematical Sciences}, 38, 777-800.
	
\end{thebibliography}

\end{document}